\documentclass[journal, doublecolumn]{IEEEtran}

\usepackage{amsmath, amsfonts, amssymb, bm, dsfont}
\usepackage{amsthm}
\usepackage{optidef}
\usepackage{graphicx}
\usepackage{float}
\usepackage{subfigure}
\usepackage{caption}
\usepackage{multirow}
\usepackage{makecell}
\usepackage{longtable}
\usepackage{diagbox}
\usepackage{changepage}
\usepackage{stfloats}
\usepackage{algorithm}
\usepackage{algorithmic}
\usepackage{cite}
\usepackage[hyphens]{url}
\usepackage{hyperref}
\hypersetup{
    breaklinks=true,
    colorlinks=true,
    linkcolor=black,
    citecolor=black,
    urlcolor=blue,
}
\usepackage{color}
\usepackage{fixltx2e}

\theoremstyle{definition}

\theoremstyle{remark}

\theoremstyle{plain}
\newtheorem{theorem}{Theorem}
\newtheorem{lemma}{Lemma}
\newtheorem{proposition}{Proposition}
\newtheorem{corollary}{Corollary}

\begin{document}
\title{On the Performance of Pinching-Antenna Systems (PASS) with Orthogonal and Non-Orthogonal Multiple Access}

\author{Yanyu~Cheng,~\IEEEmembership{Senior Member,~IEEE},
Chongjun~Ouyang,~\IEEEmembership{Member,~IEEE},
Yuanwei~Liu,~\IEEEmembership{Fellow,~IEEE},
and~George~K.~Karagiannidis,~\IEEEmembership{Fellow,~IEEE}

\thanks{
Yanyu~Cheng is with the School of Cyberspace, Hangzhou Dianzi University, Hangzhou 310018, China (e-mail: yycheng@hdu.edu.cn).

Chongjun~Ouyang is with the School of Electronic Engineering and Computer Science, Queen Mary University of London, London E1 4NS, U.K. (e-mail: c.ouyang@qmul.ac.uk).

Yuanwei~Liu is with the Department of Electrical and Electronic Engineering, The University of Hong Kong, Hong Kong (email: yuanwei@hku.hk).

George~K.~Karagiannidis is with the Department of Electrical and Computer Engineering, Aristotle University of Thessaloniki, Thessaloniki 541 24, Greece (e-mail: geokarag@auth.gr).
}}
\maketitle

\begin{abstract}
This paper conducts a comprehensive performance analysis for pinching-antenna systems (PASS) under both orthogonal multiple access (OMA) and non-orthogonal multiple access (NOMA) transmission.
Given the cost of waveguides, we consider a scenario where the waveguide is not deployed in all rooms, i.e., some users are beyond the line-of-sight (LoS) link service area of the PASS.
Specifically, we consider a PASS where a pinching antenna in one room serves two users located in separate rooms.
The wireless transmissions between the pinching antenna and the users are performed via LoS and non-line-of-sight (NLoS) links, respectively.
Closed-form expressions for the outage probabilities (OPs) of the two users are derived for the considered system.
Furthermore, asymptotic analyses in the high signal-to-noise ratio (SNR) regime are performed to reveal the achievable diversity orders.
Numerical simulations validate the accuracy of the theoretical analysis and show that: 1) compared with conventional antenna systems (CASS), the OP of the LoS user in PASS is significantly reduced for both OMA and NOMA schemes in the middle SNR regime and approaches zero as the SNR increases; 2) since the diversity orders of the NLoS user in CASS and PASS are the same, the movement of the pinching antenna has no significant effect on the OP of the NLoS user for the OMA and NOMA scenarios.
\end{abstract}


\IEEEpeerreviewmaketitle

\section{Introduction}
Flexible antenna systems, such as fluid and movable antenna systems, have received considerable attention in recent years. These novel antenna types and array architectures not only intelligently reconfigure wireless channel conditions by adjusting the antenna positions, but also provide flexibility in controlling the effective end-to-end channel gain. As the most prominent example of flexible antenna systems, the reconfigurable intelligent surface (RIS) is typically deployed between the transmitter and receiver, where it achieves physical beamforming by adjusting the phase shifts of its passive reflective elements, thereby further extending the coverage of the signal source \cite{cheng2021downlink}.

Although existing flexible antenna systems have demonstrated the ability to improve communication performance in multiple dimensions, they are still subject to significant practical limitations \cite{liu2025pinching}. First, traditional fluid antennas and movable antennas limit antenna movement to a few wavelengths, which has an insignificant impact on large-scale path loss. Therefore, their ability to combat large-scale path loss is limited. Second, in existing flexible antenna systems, it is impractical to adjust the number of active antennas therein once the communication systems are established.

Due to the aforementioned issues, pinching-antenna systems (PASS) were proposed by NTT DOCOMO as an innovative alternative \cite{suzuki2022pinching}. PASS consist of a special transmission line called a dielectric waveguide. Similar to leaky wave antennas, electromagnetic waves are emitted within the waveguide by pinching small separate dielectric particles at specific points along its length \cite{liu2025pinching}. These dielectrics are typically attached to the tips of plastic clips or clothespins, forming a structure known as a pinching antenna (PA). The waveguide can be selectively activated at various points to create targeted coverage in specific areas, allowing highly flexible and dynamic antenna placement. When compression is released at a given location, signal emission stops immediately, similar to removing a clothespin from a clothesline.

In essence, PASS can be viewed as a specific implementation of the fluid-antenna \cite{wong2020fluid,wong2021fluid,wong2020performance} or movable-antenna \cite{jin2025general,zhu2025movable, shao20256d} concepts, which provides a more flexible and scalable solution than traditional architectures. In recognition of DOCOMO's original contribution \cite{suzuki2022pinching}, we refer to this technology as PASS throughout this paper. These unique properties have motivated several recent research studies.
Moreover, PASS is consistent with the emerging vision of surface-wave communication super-highways, which aims to reduce path loss and improve signal power delivery by leveraging in-waveguide propagation through reconfigurable waveguides \cite{wong2020vision}.

\subsection{Related Work}
\subsubsection{Flexible Antennas}
According to \cite{kirtania2020flexible}, flexible antennas offer advantages over conventional antennas such as lightweight design, flexibility, and adaptability to complex surfaces.
In \cite{yang2025flexible}, the authors showed that flexible antenna arrays (FAA) increase the degrees of freedom (DoF) by adjusting the position and orientation of the antennas, allowing them to better adapt to different communication scenarios.
In \cite{wong2020fluid}, the authors proposed the innovative concept of a fluid antenna system (FAS) and demonstrated theoretically that the FAS can outperform conventional multi-antenna maximum ratio combining (MRC) systems.
In \cite{wong2021fluid}, the authors introduced the fluid antenna multiple access (FAMA) architecture. Through rigorous mathematical analysis, it was shown that FAMA can significantly increase system capacity and support a larger number of users in spatially constrained environments.
In \cite{wong2020performance}, the authors investigated the capacity performance and channel dynamics of a FAS under extremely compact spatial conditions, providing theoretical foundations for efficient communication in future miniaturized devices.
As the most prominent example of flexible antenna systems, RIS technology improves spectral efficiency \cite{cheng2021non,cheng2023performance}. In \cite{wu2019intelligent}, it is demonstrated that a large number of passive elements in a RIS can independently reflect incident signals and adjust their phase, achieving passive beamforming without the need for radio frequency (RF) chains.

\subsubsection{Pinching Antennas}
When dealing with large-scale path loss, PA is emerging as a new research direction. In PASS, under specific scenarios, adjusting the antenna’s position can transform a non-line-of-sight (NLoS) link into a line-of-sight (LoS) link, thereby enhancing communication quality. This flexibility, combined with the structural simplicity and ease of implementation, made PASS a highly practical and promising approach for real-world deployment.

In \cite{docomo2021pinching}, the authors presented typical application scenarios of the PA, demonstrating its ability to mitigate signal blockage caused by obstacles and to establish communication links in NLoS environments.
In \cite{xu2025rate}, the authors analyzed the impact of antenna positioning on channel gain, path loss, and phase shift and proposed strategies to maximize downlink transmission efficiency and enhance the received signal strength at the user end.
In \cite{tyrovolas2025performance}, the authors derived closed-form expressions for the OP and ergodic rate of PASS, revealing that waveguide loss in extended configurations had a substantial impact on performance.
In \cite{ouyang2025array}, the authors analytically derived the array gain achievable by PASS and demonstrated that optimizing the number and spacing of antennas resulted in enhanced performance.

In \cite{tegos2025minimum}, the authors formulated the optimization problem for multi-user PA systems, jointly optimizing the PA positioning and resource allocation to maximize the minimum data rate across all users in the system.
In \cite{hou2025performance}, the authors optimized the positioning of multiple PA serving a single user, and showed that the optimal placement was asymmetrical and non-uniform, significantly enhancing the system’s average sum capacity.

\subsubsection{PASS with OMA and NOMA}
From a performance analysis perspective, the authors investigated the scenario of a single PA and a single waveguide to address the key question of whether increasing the number of PAs can improve system performance in \cite{ding2025flexible}. By introducing orthogonal multiple access (OMA) and incorporating multiple PAs, they showed that user data rates could be significantly improved in a time-division multiple access (TDMA) system.
Similarly, in \cite{yang2025pinching}, the authors focused on a single PA system serving multiple users using OMA. Through performance analysis, it is found that PASS provides better performance compared to conventional antenna systems when users move along the $y$-axis. The authors investigated three different scenarios: OMA with a single PA, non-orthogonal multiple access (NOMA) with a single PA, and OMA with multiple PAs, and they proposed a low-complexity solution to maximize user data rates under these settings \cite{xie2025low}.

The flexibility of the PA allows the establishment of LoS links, thereby mitigating the effects of channel state information (CSI) imperfections, which can enhance NOMA performance.
In \cite{ding2025flexible}, the authors extended NOMA to a single waveguide, multi-PASS by using superposition coding to transmit superimposed user signals.
Theoretical analysis showed that the performance of a PASS using NOMA exceeds that of an antenna system using OMA.
Similarly, in \cite{yang2025pinching}, the authors discussed the advantages of NOMA-based PASS in multi-user environments. It was emphasized that when the number of waveguides is smaller than the number of users, the introduction of NOMA significantly improves spectrum efficiency. In addition, PA's ability to dynamically adapt to channel conditions facilitates NOMA's power allocation strategy.
To further optimize NOMA-assisted PASS, a low-complexity antenna activation algorithm based on matching theory was proposed in \cite{wang2025antenna} by converting the antenna positions from continuous variables to discrete variables to maximize the system sum rate. They formulated the antenna activation problem as a one-dimensional one-to-one matching problem and introduced a low-complexity iterative algorithm based on matching theory to maximize the system sum rate.

\subsection{Motivation and Contributions}
Recent pioneering research has demonstrated that the performance of PASS is superior to that of conventional fixed-position antenna systems for both OMA and NOMA scenarios \cite{ding2025flexible,yang2025pinching,xie2025low}.
It should be noted that the performance analysis of new technologies is crucial, and most current studies focus on the analysis of the ergodic rate of adaptive-rate systems \cite{wang2025antenna,hou2025performance,ding2025flexible}, lacking the analysis of the outage probability (OP) of fixed-rate systems.
Furthermore, due to the cost of waveguides, it is not possible to implement waveguides in all rooms, i.e., some users are beyond the LoS link service area of PASS.
This motivates us to perform a comprehensive outage performance analysis of the PASS with OMA and NOMA, where not all users are within the LoS link service area.

Specifically, we study a three-dimensional (3D) indoor communication system, where a room is equipped with a dielectric waveguide installed on the ceiling, as shown in Figure \ref{fig_system_model}. A base station (BS) serves two users located in separate and mutually obstructed rooms, which have different LoS and NLoS propagation characteristics. We analyze the system performance under two multiple access techniques: OMA and NOMA. The main contributions of this study are summarized as follows:
\begin{itemize}
    \item A comprehensive performance analysis is performed for CASS and PASS with NOMA and OMA scenarios. In particular, the OPs and diversity orders are analyzed for all cases.
    \item Considering the uniform distribution of users and the small-scale fading of the NLoS path, we derive the closed-form expressions for the OPs of all cases.
    \item To gain further insight, we derive the asymptotic approximations of the OPs at high signal-to-noise ratio (SNR) to obtain diversity orders.
    \item The analytical and simulation results show that the PASS exhibit superior performance compared to CASS under both OMA and NOMA scenarios. In particular, compared to CASS, the performance of the LoS user in PASS is significantly improved, while the performance of the NLoS user is slightly degraded.
\end{itemize}

The structure of this paper is as follows: In Section \ref{Sec_System_model}, the system model is introduced systematically. In Section \ref{Sec_CASS}, we analyze the performance of CASS, and in Section \ref{Sec_PASS}, we further analyze the performance of PASS.
Numerical results are given in Section \ref{Sec_Numerical_Results}. Finally, in Section \ref{Sec_Conclusion}, the conclusion of this work is drawn.

\section{System Model}\label{Sec_System_model}
\begin{figure}
\centering
\includegraphics[width=\linewidth]{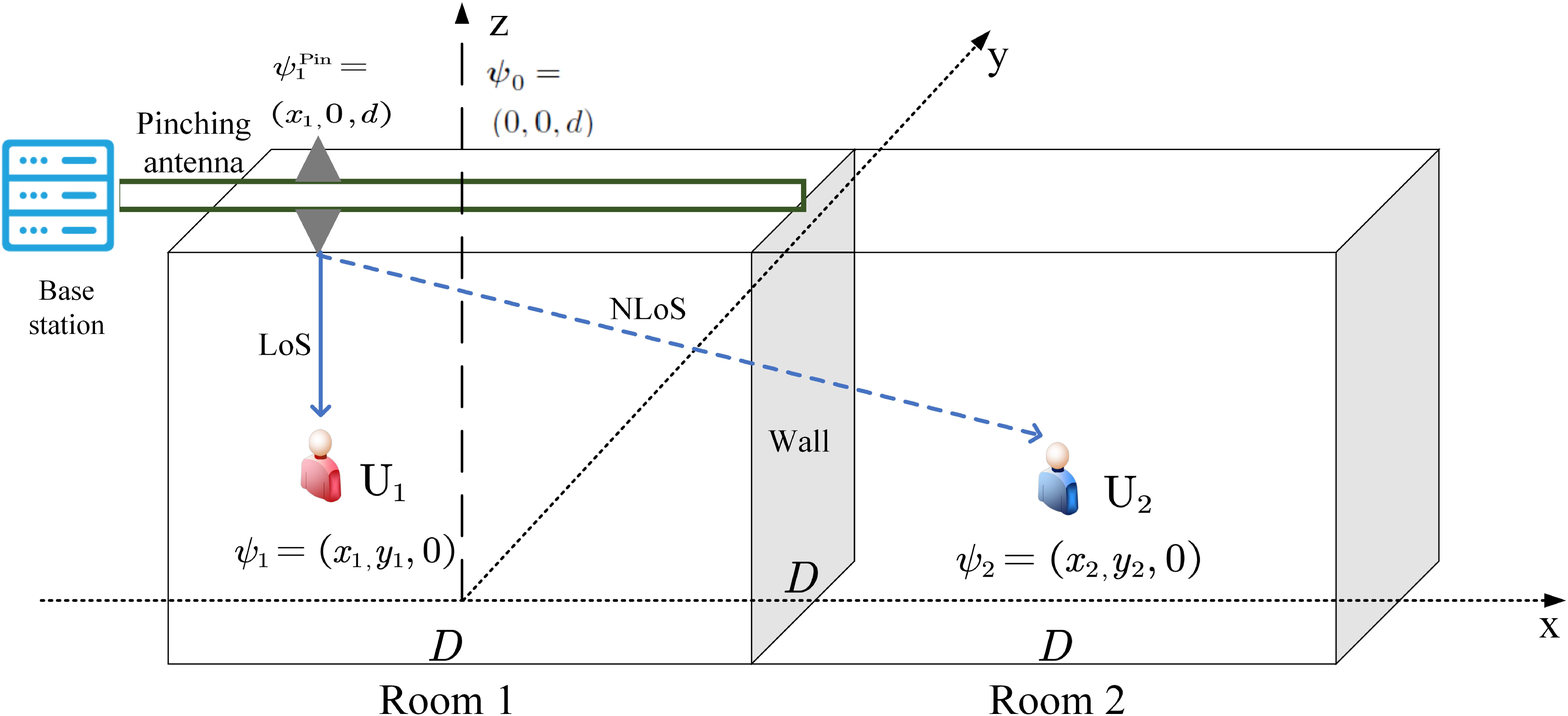}
\caption{The considered PASS with two users.}
\label{fig_system_model}
\end{figure}
Consider a wireless communication system in which a BS serves two single-antenna users, denoted by $\mathrm{U}_1$ and $\mathrm{U}_2$.
$\mathrm{U}_1$ and $\mathrm{U}_2$ are uniformly distributed in room $1$ and room $2$, respectively, where both rooms are squares with side length $D$.
To conveniently describe the system model, we establish a three-dimensional Cartesian coordinate system as shown in Fig. \ref{fig_system_model}.
$\mathrm{U}_1$ and $\mathrm{U}_2$ are located at $\psi_1=[x_1, y_1, 0]$ and $\psi_2=[x_2, y_2, 0]$, respectively.
A waveguide is placed in room~$1$ and parallel to the $x$-axis, where the height of the waveguide is denoted by $d$.
Therefore, a PA deployed on the waveguide can serve $\mathrm{U}_1$ with an LoS link and $\mathrm{U}_2$ with an NLoS link.

\subsection{Conventional Antenna Systems}
\subsubsection{CASS With OMA}
Without loss of generality, TDMA is used as an example of OMA, i.e., $\mathrm{U}_m$ is served in time slot $m$ ($m \in \{1, 2\}$).
The conventional antenna (CA) lacks installation flexibility, and it has to be deployed at a fixed location. For the considered scenario, a straightforward choice is to deploy the antenna at the center of the square, i.e., $\psi_0=[0, 0, d]$ shown in Fig. \ref{fig_system_model}, where $d$ denotes the height of the antenna. Therefore, according to the spherical-wave channel model\cite{ding2025flexible}, $\mathrm{U}_1$’s data rate is given by
\begin{equation}
\begin{split}
R_1^{Conv, OMA} = \frac{1}{2}\log_2\left(1+\frac{\eta P_1}{|\mathbf{\psi}_0-\mathbf{\psi}_1|^2 \sigma^2}\right),
\end{split}
\end{equation}
where $\eta = \frac{c^2}{16\pi^2f_c^2}$, $c$
denotes the speed of light, $f_c$ is the carrier frequency, $P_1$ is the transmit power for $\mathrm{U}_1$'s signal, $|\mathbf{\psi}_0-\mathbf{\psi}_m|$ denotes the distance between the CA and $\mathrm{U}_m$, $n_1$ and $n_2$ denote the additive white Gaussian noises (AWGNs) at $\mathrm{U}_1$ and $\mathrm{U}_2$, respectively, with the same variance $\sigma^2$.
On the other hand, $\mathrm{U}_2$’s data rate is given by
\begin{equation}
\begin{split}
R_2^{Conv, OMA} = \frac{1}{2}\log_2\left(1+\frac{P_2 |h_2|^2}{|\mathbf{\psi}_0-\mathbf{\psi}_2|^\alpha \sigma^2 }\right),
\end{split}
\end{equation}
where $P_2$ is the transmit power for $\mathrm{U}_2$'s signal, $h_2$ is the small-scale fading following the Rayleigh fading model, $\alpha$ is the path-loss exponent.

\subsubsection{CASS With NOMA}
The BS transmits the signal $x=\sqrt{\alpha_1 P_b}s_{1}+\sqrt{\alpha_2 P_b}s_{2}$, where $P_b$ denotes the transmit power at the BS, $s_{1}$ and $s_{2}$ denote the transmitted signals to $\mathrm{U}_1$ and $\mathrm{U}_2$, respectively, $\alpha_1$ and $\alpha_2$ denote the power allocation coefficients for $\mathrm{U}_1$ and $\mathrm{U}_2$, respectively ($\alpha_1+ \alpha_2=1$).
Note that we assume the fixed power allocation between two users and set $\alpha_1 < \alpha_2$ for user fairness\cite{cheng2021non}.
The received signals at $\mathrm{U}_1$ and $\mathrm{U}_2$ are given by
\begin{equation}
\begin{split}
y_1^{Conv, NOMA}=\frac{\sqrt{\eta}}{|\mathbf{\psi}_0-\mathbf{\psi}_1|} x+n_1,
\end{split}
\end{equation}
and
\begin{equation}
\begin{split}
y_2^{Conv, NOMA}=\frac{|h_2|}{|\mathbf{\psi}_0-\mathbf{\psi}_2|^{\alpha/2}} x +n_2,
\end{split}
\end{equation}
respectively,
where $n_1$ and $n_2$ denote the AWGNs at $\mathrm{U}_1$ and $\mathrm{U}_2$, respectively, with the same variance $\sigma^2$.

At $\mathrm{U}_1$, the signal of $\mathrm{U}_2$ is detected first, and the corresponding data rate is given by
\begin{equation}
\begin{split}
R_{1,2}^{Conv, NOMA}=\log_2\left(1 + \frac{\alpha_2 \rho \eta |\mathbf{\psi}_0-\mathbf{\psi}_1|^{-2}}{\alpha_1 \rho \eta |\mathbf{\psi}_0-\mathbf{\psi}_1|^{-2} + 1}\right),
\end{split}
\end{equation}
where $\rho=P_b/\sigma_n^2$ denotes the transmit SNR at the BS.
After implementing the successive interference cancellation (SIC), the signal of $\mathrm{U}_1$ is decoded, and the corresponding data rate is given by
\begin{equation}
\begin{split}
R_1^{Conv, NOMA}=\log_2\left(1 + \alpha_1 \rho \eta |\mathbf{\psi}_0-\mathbf{\psi}_1|^{-2}\right).
\end{split}
\end{equation}
At $\mathrm{U}_2$, its signal is decoded directly by regarding $\mathrm{U}_1$'s signal as interference, and its data rate is given by
\begin{equation}
\begin{split}
R_2^{Conv, NOMA}=\log_2\left(1 + \frac{\alpha_2 \rho |h_2|^2 |\mathbf{\psi}_0-\mathbf{\psi}_2|^{-\alpha}}{\alpha_1 \rho |h_2|^2 |\mathbf{\psi}_0-\mathbf{\psi}_2|^{-\alpha} + 1}\right).
\end{split}
\end{equation}

\subsection{Pinching-Antenna Systems}
\subsubsection{PASS With OMA}
The key feature of PAs is their deployment flexibility. They can be moved on a scale much larger than a wavelength and deployed at the positions closest to the users.
This paper considers the case of a single PA. During the $m$-th time slot, $\mathrm{U}_m$ is served, and the PA is moved to the position closest to $\mathrm{U}_m$, denoted by $\mathbf{\psi}_m^{Pin}$.
\begin{equation}
\begin{split}
R_1^{Pin,OMA} = \frac{1}{2}\log_2\left(1+\frac{\eta P_1}{|\mathbf{\psi}_1^{Pin}-\mathbf{\psi}_1|^2 \sigma^2}\right),
\end{split}
\end{equation}
where $\eta = \frac{c^2}{16\pi^2f_c^2}$, $P_1$ is the transmit power for $\mathrm{U}_1$'s signal, and $|\mathbf{\psi}_1^{Pin}-\mathbf{\psi}_m|$ denotes the distance between the PA and $\mathrm{U}_m$.
On the other hand, $\mathrm{U}_2$’s data rate is given by
\begin{equation}
\begin{split}
R_2^{Pin,OMA} = \frac{1}{2}\log_2\left(1+\frac{P_2 |h_2|^2}{|\mathbf{\psi}_1^{Pin}-\mathbf{\psi}_2|^\alpha \sigma^2}\right),
\end{split}
\end{equation}
where $P_2$ is the transmit power for $\mathrm{U}_2$'s signal, $h_2$ is the small-scale fading following the Rayleigh fading model, $\alpha$ is the path-loss exponent.

\subsubsection{PASS With NOMA}
At $\mathrm{U}_1$, the signal of $\mathrm{U}_2$ is detected first, and the corresponding data rate is given by
\begin{equation}
\begin{split}
R_{1,2}^{Pin,NOMA}=\log_2\left(1 + \frac{\alpha_2 \rho \eta |\mathbf{\psi}_1^{Pin}-\mathbf{\psi}_1|^{-2}}{\alpha_1 \rho \eta |\mathbf{\psi}_1^{Pin}-\mathbf{\psi}_1|^{-2} + 1}\right).
\end{split}
\end{equation}
After implementing the successive interference cancellation (SIC), the signal of $\mathrm{U}_1$ is decoded, and the corresponding data rate is given by
\begin{equation}
\begin{split}
R_1^{Pin,NOMA}=\log_2\left(1 + \alpha_1 \rho \eta |\mathbf{\psi}_1^{Pin}-\mathbf{\psi}_1|^{-2}\right).
\end{split}
\end{equation}
At $\mathrm{U}_2$, its signal is decoded directly by regarding $\mathrm{U}_1$'s signal as interference, and its data rate is given by
\begin{equation}
\begin{split}
R_2^{Pin,NOMA}=\log_2\left(1 + \frac{\alpha_2 \rho |h_2|^2 |\mathbf{\psi}_1^{Pin}-\mathbf{\psi}_2|^{-\alpha}}{\alpha_1 \rho |h_2|^2 |\mathbf{\psi}_1^{Pin}-\mathbf{\psi}_2|^{-\alpha} + 1}\right).
\end{split}
\end{equation}

\section{Performance Analysis for CASS}\label{Sec_CASS}
In this section, we investigate the performance of CASS with OMA and NOMA. 
To realize this, the distributions of $|\psi_{0} - \psi_{1} |^{2}$ and $|\psi_{0} - \psi_{2} |^{2}$ are required.

\subsection{Distance Statistics for CASS}
\begin{lemma}\label{lemma-z1}
Denote that \( Z_{1} = | \psi_{0} - \psi_{1} |^{2} \), where $\psi_0 = [0, 0, d]$, \(\psi_1 = [x_1, y_1, 0]\), \(x_{1} \in \left[-\frac{D}{2}, \frac{D}{2}\right]\) and \(y_{1} \in \left[-\frac{D}{2}, \frac{D}{2}\right]\). Its probability density function (PDF) and cumulative distribution function (CDF) are given by
\begin{equation}\label{PDF-z1}
\begin{split}
f_{Z_{1}}(z) = 
\begin{cases} 
0, &z < d^2, \\ 
\frac{\pi}{D^2}, & d^2 \leq z < d^2 + \frac{D^2}{4}, \\ 
\frac{4 \arcsin \left( \frac{D}{2\sqrt{\zeta}}  - \pi \right)}{D^2}, & d^2 + \frac{D^2}{4} \leq z < d^2 + \frac{D^2}{2}, \\
0, & d^2 + \frac{D^2}{2} \leq z, 
\end{cases}
\end{split}
\end{equation}
and
\begin{equation}\label{CDF-Z1}
\begin{split}
F_{Z_{1}}(z) = 
\begin{cases} 
0, & z < d^2, \\ 
\frac{\pi \zeta}{D^2}, & d^2 \leq z < d^2 + \frac{D^2}{4}, \\ 
\frac{\hbar - \pi\zeta}{D^2}, & d^2 + \frac{D^2}{4} \leq z < d^2 + \frac{D^2}{2}, \\
1, & d^2 + \frac{D^2}{2} \leq z,
\end{cases}
\end{split}
\end{equation}
respectively, where $\zeta=z-d^2$, $\epsilon=z-d^2-\frac{D^2}{4}$, and $\hbar = 2D\sqrt{\epsilon} + 4\zeta \arcsin \left( \frac{D}{2\sqrt{\zeta}} \right)$.
\end{lemma}

\begin{lemma}\label{lemma-z2}
Denote that \( Z_{2} = | \psi_{0} - \psi_{2} |^{2} \), where \(\psi_0 = [0, 0, d]\), \(\psi_2 = [x_2, y_2, 0]\), \(x_{2} \in \left[\frac{D}{2}, \frac{3D}{2}\right]\), and \(y_{2} \in \left[-\frac{D}{2}, \frac{D}{2}\right]\). Its PDF and CDF are given by
\begin{equation}
\label{Ppf-z2}
\begin{split}
f_{Z_{2}}(z) = 
\begin{cases} 
0, &z < d^2+\frac{D^2}{4}, \\ 
\frac{\frac{\pi}{2}-\arcsin \left( \frac{D}{2\sqrt{\zeta}} \right)}{D^2}, &  d^2+\frac{D^2}{4}  \leq z < d^2 + \frac{D^2}{2}, \\ 
\frac{\frac{\pi}{2}-\arcsin \left( \sqrt{1-\frac{D^2}{4\zeta}} \right)}{D^2}, & d^2 + \frac{D^2}{2} \leq z < d^2 + \frac{9D^2}{4}, \\
\frac{\imath}{D^2}, & d^2 + \frac{9D^2}{4} \leq z < d^2 + \frac{5D^2}{2}, \\
0, & d^2 + \frac{5D^2}{2} \leq z, 
\end{cases}
\end{split}
\end{equation}
and
\begin{equation}\label{CDF-Z2}
\begin{split}
F_{Z_{2}}(z) = 
\begin{cases} 
0, &z < d^2+\frac{D^2}{4}, \\ 
\frac{1}{D^2}\left(\frac{\pi\zeta}{2} - \frac{\hbar}{4}\right),  &d^2+\frac{D^2}{4}  \leq z < d^2 + \frac{D^2}{2}, \\ 
\frac{\jmath}{D^2},&d^2 + \frac{D^2}{2} \leq z < d^2 + \frac{9D^2}{4}, \\
\frac{\imath}{2D^2},&d^2 + \frac{9D^2}{4} \leq z < d^2 + \frac{5D^2}{2} ,\\
1,& d^2 + \frac{5D^2}{2} \leq z ,
\end{cases}
\end{split}
\end{equation}
respectively, where $\imath = D\sqrt{\epsilon}-D^2+3D\sqrt{\epsilon-2D^2}+2\zeta \left(\arcsin\left( \frac{D}{2\sqrt{\zeta}} \right)+\arcsin\left(\frac{3D}{2\sqrt{\zeta}}\right)-\frac{\pi}{2}\right)$ and $\jmath =\frac{D\sqrt{\epsilon}-D^2+\pi\zeta}{2}-\zeta\arcsin \left( \sqrt{1-\frac{D^2}{4\zeta}} \right)$.
\end{lemma}

\subsection{CASS With OMA}
The OPs of  $\mathrm{U}_{1}$ and $\mathrm{U}_{2}$ are given by
\begin{equation}\label{Conv-U1-OMA}
\begin{split}
 \mathbb{P}_{1}^{Conv,OMA}&=\Pr( R_{1}^{Conv,OMA}<\bar{R}),
\end{split}
\end{equation}
and
\begin{equation}\label{Conv-U2-OMA}
\begin{split}
\mathbb{P}_{2}^{Conv,OMA}&=\Pr( R_{2}^{Conv,OMA}<\bar{R}),
\end{split}
\end{equation}
respectively, where $\bar{R}$ is the target date rate.

Based on the distance statistics in Lemma \ref{lemma-z1} and Lemma \ref{lemma-z2}, the OPs of $\mathrm{U}_{1}$ and $\mathrm{U}_{2}$ are derived in the following theorem.
\begin{theorem}\label{theorem-conv-oma}
In CASS with OMA, the OPs of $\mathrm{U}_{1}$ and $\mathrm{U}_{2}$ are given by
\begin{equation}\label{Conv-U1-OMA-OP}
\begin{split}
 \mathbb{P}_{1}^{Conv,OMA}&=1-\mathbb{I}_{1},
\end{split}
\end{equation}
and
\begin{equation}\label{Conv-U2-OMA-OP}
\begin{split}
 \mathbb{P}_{2}^{Conv,OMA}&= 1-\frac{1}{D^2} (\mathbb{J}_{1} +\mathbb{Q}_{1}  + \mathbb{K}_{1}),
\end{split}
\end{equation}
respectively, where $\mathbb{I}_{1}=F_{Z_{1}}(a)$, $a=\frac{\eta \rho}{2^{M\bar{R}}-1}$, $\mathbb{J}_{1}=\sum_{i=1}^{n} \omega_i \mathbf{j}_{1}(t_i)$,  $\mathbb{Q}_{1}=\sum_{i=1}^{n} \omega_i \mathbf{q}_{1}(t_i)$, $\mathbb{K}_{1}=\sum_{i=1}^{n} \omega_i \mathbf{k}_{1}(t_i)$, $n$ is the number of nodes for the Chebyshev-Gauss quadrature\footnote{$n$ is a trade-off parameter of accuracy and complexity. The approximation accuracy increases as $n$ increases.},
$\omega_i = \frac{\pi}{n}$ is the weight, $t_i = \cos \left( \frac{2i - 1}{2n} \pi \right)$, 
$\mathbf{j}_{1}(t) = \frac{D^2}{8} e^{-\varpi_{1} (\frac{D^2}{8} t + \frac{3D^2}{8})^{\frac{\alpha}{2}}} \left( \frac{\pi}{2} - \arcsin \left( \sqrt{\frac{2}{t + 3}} \right) \right) \sqrt{1 - t^2}$, 
$\mathbf{q}_{1}(t)=\frac{7D^{2}}{8} e^{-\varpi_{1}\left(\frac{7D^{2}}{8}t + \frac{11D^{2}}{8}\right)^{\frac{\alpha}{2}}} \Big(\frac{\pi}{2} - \arcsin\left(\sqrt{\frac{7t + 9}{7t + 11}}\right)\Big)\\
\times \sqrt{1 - t^{2}}$,
$\mathbf{k}_{1}(t)=\frac{D^{2}}{8} e^{-\varpi_{1}\left(\frac{D^{2}}{8}t + \frac{19D^{2}}{8}\right)^{\frac{\alpha}{2}}}\Big(\arcsin\left(3\sqrt{\frac{2}{t + 19}}\right)\\
- \frac{\pi}{2} +  \arcsin\left(\sqrt{\frac{2}{t + 19}}\right) \Big)  \sqrt{1 - t^{2}}$, $ \varpi_{1}=\frac{(2^{M\bar{R}}-1)}{\rho}$, and we assume that $\rho=P_1/\sigma^2=P_2/\sigma^2$.
\end{theorem}

\begin{proof}\label{Conv-proof1}
See Appendix \ref{Appen-thro-1}.
\end{proof}

\begin{proposition}\label{proposition-conv-oma}
In the high-SNR regime, $\mathbb{P}_{1}^{Conv,OMA}$ can be approximated as
\begin{equation}\label{eq_conv-oma-high-snr-u1}
    \begin{split}
        \mathbb{P}_{1}^{Conv,OMA,\infty} = 0,
    \end{split}
\end{equation} 
and $\mathbb{P}_{2}^{Conv,OMA}$ can be approximated as
\begin{equation}\label{eq_conv-oma-high-snr-u2}
    \begin{split}
        \mathbb{P}_{2}^{Conv,OMA,\infty} = \tilde{c}_1\rho^{-1},
    \end{split}
\end{equation} 
where $\tilde{c}_1 = \int_{0}^{\infty}
            \omega_{1}z^{\frac{\alpha}{2}}f_{Z_{2}}(z)dz$ and $\omega_1 = 2^{M\bar{R}}-1$.
Specifically, $\mathbb{P}_{1}^{Conv,OMA}$ becomes zero for $\rho \geq \frac{(\frac{D^2}{2}+d^2)(2^{M\bar{R}}-1)}{\eta}$.
\end{proposition}

\begin{proof}\label{Conv-OMA-high-SNR-proof}
When $\rho \to \infty$, $F_{Z_{1}}(\infty) = 1$, and we have $\mathbb{P}_{1}^{Conv,OMA,\infty} = 0$. By expanding the exponential function in $\mathbb{P}_{2}^{Conv,OMA}= 1 - \int_{0}^{\infty} e^{-\varpi_{1} z^{\frac{\alpha}{2}}} f_{Z_{2}}(z) dz$, we have
    \begin{equation}\label{Conv-OMA-high-SNR-proof-r2}
        \begin{split}
            \mathbb{P}_{2}^{Conv,OMA} &= 1 - \sum_{n=0}^{\infty}\rho^{-n}\int_{0}^{\infty}
            \frac{(-\omega_{1})^n}{n!}z^{\frac{n\alpha}{2}}f_{Z_{2}}(z)dz \\ 
            &= \sum_{n=1}^{\infty}\rho^{-n}\int_{0}^{\infty}
            \frac{(-1)^{n+1}\omega_{1}^n}{n!}z^{\frac{n\alpha}{2}}f_{Z_{2}}(z)dz.
        \end{split}
    \end{equation}
After extracting the dominant term, we can obtain \eqref{eq_conv-oma-high-snr-u2}.
The proof is complete.
\end{proof}

\begin{corollary}\label{corollary-conv-oma}
In the considered CASS with OMA, the diversity order $\mathrm{U}_2$ is given by  $\mathcal{D}_{2}^{Conv,OMA} = 1$.
\end{corollary}

\subsection{CASS With NOMA}
According to the distance statistic, the OPs of $ \mathrm{U}_{1}$ and $\mathrm{U}_{2}$ are given by 
\begin{equation}\label{Conv-U1-NOMA}
\begin{split}
 \mathbb{P}_{1}^{Conv,NOMA}&=\Pr( R_{1}^{Conv,NOMA}<\bar{R}),
\end{split}
\end{equation}
and
\begin{equation}\label{Conv-U2-NOMA}
\begin{split}
\mathbb{P}_{2}^{Conv,NOMA}&=\Pr( R_{2}^{Conv,NOMA}<\bar{R}),
\end{split}
\end{equation}
respectively.
\begin{theorem}\label{theorem-conv-noma}
In CASS with NOMA, the OPs of $ \mathrm{U}_{1}$ and $\mathrm{U}_{2}$ are  given by
\begin{equation}\label{Conv-U1-NOMA-OP}
\begin{split}
\mathbb{P}_{1}^{Conv,NOMA}&=1-\mathbb{I}_{2},
\end{split}
\end{equation}
and
\begin{equation}\label{Conv-U2-NOMA-OP}
\begin{split}
 \mathbb{P}_{2}^{Conv,NOMA}&= 1-\frac{1}{D^2} (\mathbb{J}_{2} +\mathbb{Q}_{2}  + \mathbb{K}_{2}),
\end{split}
\end{equation}
respectively, where $\mathbb{I}_{2}=F_{Z_{1}}(b)$, $b=\frac{\alpha_{1}\eta \rho}{(2^{\bar{R}}-1)}$, $\mathbb{J}_{2}=\sum_{i=1}^{n} \omega_i \mathbf{j}_{2}(t_i)$,  $\mathbb{Q}_{2}=\sum_{i=1}^{n} \omega_i \mathbf{q}_{2}(t_i)$, $\mathbb{K}_{2}=\sum_{i=1}^{n} \omega_i \mathbf{k}_{2}(t_i)$, n is the number of nodes for the Chebyshev-Gauss quadrature, $\omega_i = \frac{\pi}{n}$ is the weight, $t_i = \cos \left( \frac{2i - 1}{2n} \pi \right)$,
$\mathbf{j}_{2}(t) = \frac{D^2}{8} e^{-\varpi_{2} (\frac{D^2}{8} t + \frac{3D^2}{8})^{\frac{\alpha}{2}}} \left( \frac{\pi}{2} - \arcsin \left( \sqrt{\frac{2}{t + 3}} \right) \right) \sqrt{1 - t^2}$,
$\mathbf{q}_{2}(t)=\frac{7D^{2}}{8} e^{-\varpi_{2}\left(\frac{7D^{2}}{8}t + \frac{11D^{2}}{8}\right)^{\frac{\alpha}{2}}} \Big(\frac{\pi}{2} - \arcsin\left(\sqrt{\frac{7t + 9}{7t + 11}}\right)\Big) \\ \times 
\sqrt{1 - t^{2}}$,
$\mathbf{k}_{2}(t)=\frac{D^{2}}{8} e^{-\varpi_{2}\left(\frac{D^{2}}{8}t + \frac{19D^{2}}{8}\right)^{\frac{\alpha}{2}}}   \Big(\arcsin\left(3\sqrt{\frac{2}{t + 19}}\right) \\ - \frac{\pi}{2} + \arcsin\left(\sqrt{\frac{2}{t + 19}}\right) \Big) \sqrt{1 - t^{2}}$, and $ \varpi_{2}= \frac{2^{\bar{R}}-1}{\rho(\alpha_{2}-\alpha_{1}(2^{\bar{R}}-1))}$.
\end{theorem}

\begin{proof}\label{Conv-proof2}
See Appendix \ref{Appen-thro-2}.
\end{proof}

\begin{proposition}\label{proposition-conv-noma}
In the high-SNR regime, $\mathbb{P}_{1}^{Conv,NOMA}$ can be approximated as
\begin{equation}\label{eq_conv-noma-high-snr-u1}
\begin{split}
\mathbb{P}_{1}^{Conv,NOMA,\infty} = 0,
\end{split}
\end{equation} 
and $\mathbb{P}_{2}^{Conv,NOMA}$ can be approximated as
\begin{equation}\label{eq_conv-noma-high-snr-u2}
\begin{split}
\mathbb{P}_{2}^{Conv,NOMA,\infty} = \tilde{c}_2\rho^{-1},
\end{split}
\end{equation} 
where $\tilde{c}_2 = \int_{0}^{\infty}\omega_{2}z^{\frac{\alpha}{2}}f_{Z_{2}}(z)dz$ and $\omega_2=\frac{2^{M\bar{R}}-1}{\alpha_2-\alpha_1(2^{M\bar{R}}-1)}$.
Specifically, $\mathbb{P}_{1}^{Conv,NOMA}$ becomes zero for $\rho \geq \frac{(\frac{D^2}{2}+d^2)(2^{\bar{R}}-1)}{\alpha_1 \eta}$.
\end{proposition}

\begin{proof}\label{Conv-NOMA-high-SNR-proof}
    When $\rho \to \infty$, $F_{Z_{1}}(\infty) = 1$, and we have $\mathbb{P}_{1}^{Conv,OMA,\infty} = 0$. By expanding the exponential function in $\mathbb{P}_{2}^{Conv,NOMA}= 1 - \int_{0}^{\infty} e^{-\varpi_{2} z^{\frac{\alpha}{2}}} f_{Z_{2}}(z) dz$, we have
    \begin{equation}\label{Conv-NOMA-high-SNR-proof-r2}
        \begin{split}
            \mathbb{P}_{2}^{Conv,NOMA} &= 1 - \sum_{n=0}^{\infty}\rho^{-n}\int_{0}^{\infty}
            \frac{(-\omega_{2})^n}{n!}z^{\frac{n\alpha}{2}}f_{Z_{2}}(z)dz \\ 
            &= \sum_{n=1}^{\infty}\rho^{-n}\int_{0}^{\infty}
            \frac{(-1)^{n+1}\omega_{2}^n}{n!}z^{\frac{n\alpha}{2}}f_{Z_{2}}(z)dz  
        \end{split}
    \end{equation}
    After extracting the dominant term, we can obtain \eqref{eq_conv-noma-high-snr-u2}.
The proof is complete.
\end{proof}

\begin{corollary}\label{corollary-conv-noma}
In the considered CASS with NOMA, the diversity order of $\mathrm{U}_2$ is given by $\mathcal{D}_{2}^{Conv,OMA} = 1$.
\end{corollary}

\section{Performance Analysis for PASS}\label{Sec_PASS}
In this section, we investigate the performance of the PASS with OMA and NOMA. 
To realize this, the distributions of $|\psi_{1}^{pin} - \psi_{1} |^{2}$ and $|\psi_{1}^{pin} - \psi_{2} |^{2}$ are required.

\subsection{Distance Statistics for PASS}
\begin{lemma}\label{lemma-z3}
Denote that \( Z_{3} = |\psi_{1}^{pin} - \psi_{1} |^{2} \), where \(\psi_1^{pin} = [x_{1}, 0, d]\) , \(\psi_1 = [x_1, y_1, 0]\), \(x_{1} \in \left[-\frac{D}{2}, \frac{D}{2}\right]\) and  \(y_{1} \in \left[-\frac{D}{2}, \frac{D}{2}\right]\). Its PDF and CDF are given by
\begin{equation}\label{PDF-z3}
\begin{split}
f_{Z_{3}}(z) = 
\begin{cases} 
0, &z < d^2, \\ 
\frac{1}{D\sqrt{z-d^2}}, & d^2 \leq z < d^2 + \frac{D^2}{4} ,\\ 
0, & d^2 + \frac{D^2}{4} \leq z, 
\end{cases}
\end{split}
\end{equation}
and
\begin{equation}\label{CDF-z3}
\begin{split}
F_{Z_{3}}(z) = 
\begin{cases} 
0, &z < d^2, \\ 
\frac{2\sqrt{z-d^2}}{D}, & d^2 \leq z < d^2 + \frac{D^2}{4}, \\ 
1, & d^2 + \frac{D^2}{4} \leq z, 
\end{cases}
\end{split}
\end{equation}
respectively.
\end{lemma}

\begin{lemma}\label{lemma-R2-pin}
Denote that \( Z_{4} = | \psi_{1}^{pin} - \psi_{2} |^{2} \), where \(\psi_1^{pin} = [x_1, 0, d]\) , \(\psi_2 = [x_2, y_2, 0]\), \(x_{1} \in \left[-\frac{D}{2}, \frac{D}{2}\right]\), \(x_{2} \in \left[\frac{D}{2}, \frac{3D}{2}\right]\), and  \(y_{2} \in \left[-\frac{D}{2}, \frac{D}{2}\right]\). Its PDF and CDF are given by
\begin{equation}\label{PDF-z4}
\begin{split}
f_{Z_{4}}(z) = 
\begin{cases} 
0, &z < d^2, \\ 
\frac{\sqrt{\zeta}}{D^3}, & d^2 \leq z < d^2 + \frac{D^2}{4} ,\\ 
\frac{1}{2D^2}, & d^2 + \frac{D^2}{4} \leq z < d^2 + D^2, \\ 
\ell, & d^2 + D^2 \leq z < d^2 + \frac{5D^2}{4}, \\ 
\partial, & d^2 + \frac{5D^2}{4} \leq z < d^2 + 4D^2, \\
\kappa, & d^2 + 4D^2 \leq z < d^2 + \frac{17D^2}{4} ,\\ 
0,& d^2 + \frac{17D^2}{4} \leq z, 
\end{cases}
\end{split}
\end{equation}
and
\begin{equation}\label{CDF-z4}
\begin{split}
F_{Z_{4}}(z) = 
\begin{cases} 
0, &z < d^2, \\ 
\frac{2\zeta^{\frac{3}{2}}}{3D^3} ,& d^2 \leq z < d^2 + \frac{D^2}{4} ,\\ 
\frac{\zeta}{2D^2}-\frac{1}{24} ,& d^2 + \frac{D^2}{4} \leq z < d^2 + D^2, \\ 
\varrho ,& d^2 + D^2 \leq z < d^2 + \frac{5D^2}{4}, \\ 
\tau, & d^2 + \frac{5D^2}{4} \leq z < d^2 + 4D^2 ,\\
\varsigma, & d^2 + 4D^2 \leq z < d^2 + \frac{17D^2}{4}, \\ 
1,& d^2 + \frac{17D^2}{4} \leq z, 
\end{cases}
\end{split}
\end{equation}
respectively, where $\ell = \frac{1}{2D^2}-2\frac{\sqrt{z-d^2-D^2}}{D^3} + \frac{2}{D^2}\Big( \arctan \left(\frac{\sqrt{z-d^2-D^2}}{D} \right)-\frac{D\sqrt{z-d^2-D^2}}{D^2} + \frac{D\sqrt{z-d^2-D^2}}{\zeta}\Big)$, 
$\kappa = \frac{2}{D\sqrt{4(z-d^2)-D^2}} - \frac{1}{2D^2}+\frac{\sqrt{z-d^2-4D^2}}{D^3} + \frac{2}{D^2}\Big(
\Big(\arctan \left(\frac{D}{\sqrt{4(z-d^2)-D^2}} \right) - \arctan \left(\frac{\sqrt{z-d^2-4D^2}}{2D} \right) + 
\frac{7D\sqrt{z-d^2-4D^2}}{4\zeta}\Big)
-\frac{D}{\sqrt{4(z-d^2)-D^2}}-\frac{7D\sqrt{z-d^2-4D^2}}{4\zeta} 
\Big)$, 
$\partial = \frac{2\pi-1}{2D^2}-\frac{2}{D^2}\arctan \left(\frac{\sqrt{4(z-d^2)-D^2}}{D} \right)$, 
$\varrho = \frac{2\zeta}{D^2}\Big(
\arctan \left( \frac{\sqrt{z-d^2-D^2}}{D} \right) - \frac{D\sqrt{z-d^2-D^2}}{\zeta}
\Big) + \frac{\zeta}{2D^2} - \frac{4(z-d^2-D^2)^{\frac{3}{2}}}{3D^3} - \frac{1}{24}$, 
$\tau = \frac{\zeta \pi}{D^2} - \frac{2\zeta}{D^2}\Bigg(
\arctan \left( \frac{\sqrt{4(z-d^2)-D^2}}{D} \right) - \frac{D\sqrt{4(z-d^2)-D^2}}{4\zeta}
\Bigg) - \frac{\zeta}{2D^2} - \frac{23}{24}$, and $\varsigma = 1- \Big(\frac{47}{24} - \frac{\sqrt{4(z-d^2)-D^2}}{D} + \frac{\zeta}{2D^2} - \frac{2(z-d^2-4D^2)^{\frac{3}{2}}}{3D^3}
-\frac{2\zeta}{D^2}\Big(
\arctan \left( \frac{D}{\sqrt{4(z-d^2)-D^2}} \right) -\arctan \left( \frac{\sqrt{z-d^2-4D^2}}{2D} \right)
+\frac{7D\sqrt{z-d^2-4D^2}}{4\zeta}
\Big)\Big)$.
\end{lemma}

\subsection{PASS With OMA}
The OPs of  $\mathrm{U}_{1}$ and $\mathrm{U}_{2}$ are given by
\begin{equation}\label{Pin-U1-OMA}
\begin{split}
 \mathbb{P}_{1}^{Pin,OMA}&=\Pr( R_{1}^{Pin,OMA}<\bar{R}),
\end{split}
\end{equation}
and
\begin{equation}\label{Pin-U2-OMA}
\begin{split}
\mathbb{P}_{2}^{Pin,OMA}&=\Pr( R_{2}^{Pin,OMA}<\bar{R}),
\end{split}
\end{equation}
respectively, where $\bar{R}$ is the target date rate.

Based on the distance statistics in Lemma \ref{lemma-z3} and Lemma \ref{lemma-R2-pin}, the OPs of $\mathrm{U}_{1}$ and $\mathrm{U}_{2}$  are derived in the following theorem.
\begin{theorem}\label{theorem-pin-oma}
In PASS with OMA, the OPs of $\mathrm{U}_{1}$ and $\mathrm{U}_{2}$ are given by
\begin{equation}\label{Pin-U1-OMA-OP}
\begin{split}
 \mathbb{P}_{1}^{Pin,OMA}&=1-\mathbb{I}_{3},
\end{split}
\end{equation}
and
\begin{equation}\label{Pin-U2-OMA-OP}
\begin{split}
 \mathbb{P}_{2}^{Pin,OMA}&= 1- (\mathbb{J}_{3} +\mathbb{Q}_{3}  + \mathbb{K}_{3} +\mathbb{C}_{3} +\mathbb{V}_{3}),
\end{split}
\end{equation}
respectively, where $\mathbb{I}_{3}=F_{Z_{3}}(a)$, $\mathbb{J}_{3} = \sum_{i=1}^{n} \omega_{i} \mathbf{j}_{3}(i)$, $\mathbb{Q}_{3} = \sum_{i=1}^{n} \omega_{i} \mathbf{q}_{3}(i)$, $\mathbb{K}_{3} = \sum_{i=1}^{n} \omega_{i} \mathbf{k}_{3}(i)$, $\mathbb{C}_{3} = \sum_{i=1}^{n} \omega_{i} \mathbf{c}_{3}(i)$, $\mathbb{V}_{3} = \sum_{i=1}^{n} \omega_{i} \mathbf{v}_{3}(i)$, $n$ is the number of nodes for the Chebyshev-Gauss quadrature, $\omega_i = \frac{\pi}{n}$ is the weight, $t_i = \cos \left( \frac{2i - 1}{2n} \pi \right)$, $\mathbf{j}_{3}(t) = \frac{1}{8} \sqrt{1 - t^{2}} e ^{-\varpi_1 (\frac{D^{2}t}{8}+d^{2}+\frac{D^{2}}{8}) ^{\frac{\alpha}{2}} }\frac{\sqrt{t + 1}}{2\sqrt{2}}$, 
$\mathbf{q}_{3}(t) = \frac{3}{16} \sqrt{1 - t^{2}} e ^{-\varpi_1(\frac{3D^{2}t}{8}+d^{2}+\frac{5D^{2}}{8})^{\frac{\alpha}{2}} }$, 
$\mathbf{k}_{3}(t) = \frac{1}{8} \sqrt{1 - t^{2}} e^{-\varpi_1\left(\frac{D^{2}t}{8}+d^{2}+\frac{9D^{2}}{8}\right)^{\frac{\alpha}{2}} } \Big(\frac{1}{2}-\frac{\sqrt{t + 1}} {\sqrt{2}}  +2\arctan\left(\frac{\sqrt{t + 1}}{2\sqrt{2}}\right)\Big) $, 
$c_{3}(t) = \frac{11\sqrt{1 - t^{2}}e ^{-\varpi_1\left(\frac{11D^{2}t}{8}+d^{2}+\frac{21D^{2}}{8}\right)^{\frac{\alpha}{2}} }}{8}\\ \times \Big(\frac{2\pi - 1}{2}-2\arctan\left(\sqrt{\frac{11t + 19}{2}}\right)\Big)$, 
$v_{3}(t) =  \frac{\sqrt{1 - t^{2}}e ^{-\varpi_1\left(\frac{D^{2}t}{8}+d^{2}+\frac{33D^{2}}{8}\right)^{\frac{\alpha}{2}} }}{8} \Big(-\frac{1}{2}+\frac{\sqrt{t + 1}}{2\sqrt{2}} +2\Big(\arctan\left(\sqrt{\frac{2}{t + 31}}\right) -\arctan\left(\frac{\sqrt{t + 1}}{4\sqrt{2}}\right)\Big)\Big)$, and $ \varpi_{1}=\frac{2^{M\bar{R}}-1}{\rho}$.
\end{theorem}

\begin{proof}\label{Pin-proof1}
See Appendix \ref{Appen-thro-3}.
\end{proof}

\begin{proposition}\label{proposition-pin-oma}
In the high-SNR regime, $\mathbb{P}_{1}^{Pin,OMA}$ can be approximated as
\begin{equation}\label{eq_pin-oma-high-snr-u1}
    \begin{split}
        \mathbb{P}_{1}^{Pin,OMA,\infty} = 0,
    \end{split}
\end{equation} 
and $\mathbb{P}_{2}^{Pin,OMA}$ can be approximated as
\begin{equation}\label{eq_pin-oma-high-snr-u2}
    \begin{split}
        \mathbb{P}_{2}^{Pin,OMA,\infty} = \tilde{c}_3\rho^{-1},
    \end{split}
\end{equation} 
where $\tilde{c}_3 = \int_{0}^{\infty}
            \omega_{1}z^{\frac{\alpha}{2}}f_{Z_{4}}(z)dz$.
Specifically, $\mathbb{P}_{1}^{Pin,OMA}$ becomes zero for $\rho \geq \frac{(\frac{D^4}{2}+d^2)(2^{M\bar{R}}-1)}{\eta}$.
\end{proposition}

\begin{proof}\label{pin-OMA-high-SNR-proof}
    When $\rho \to \infty$, $F_{Z_{3}}(\infty) = 1$, and we have $\mathbb{P}_{1}^{Pin,OMA,\infty} = 0$. By expanding the exponential function in $\mathbb{P}_{2}^{Pin,OMA}= 1 - \int_{0}^{\infty} e^{-\varpi_{1} z^{\frac{\alpha}{2}}} f_{Z_{4}}(z) dz$, we have
    \begin{equation}\label{pin-OMA-high-SNR-proof-r2}
        \begin{split}
            \mathbb{P}_{2}^{Pin,OMA} &= 1 - \sum_{n=0}^{\infty}\rho^{-n}\int_{0}^{\infty}
            \frac{(-\omega_{1})^n}{n!}z^{\frac{n\alpha}{2}}f_{Z_{4}}(z)dz \\ 
            &= \sum_{n=1}^{\infty}\rho^{-n}\int_{0}^{\infty}
            \frac{(-1)^{n+1}\omega_{1}^n}{n!}z^{\frac{n\alpha}{2}}f_{Z_{4}}(z)dz  
        \end{split}
    \end{equation}
After extracting the dominant term, we can obtain \eqref{eq_pin-oma-high-snr-u2}.
The proof is complete.
\end{proof}

\begin{corollary}\label{corollary-pin-oma}
In the considered PASS with OMA, the diversity order of $\mathrm{U}_2$ is given by $\mathcal{D}_{2}^{Pin,OMA} = 1$.
\end{corollary}

\subsection{PASS With NOMA}
According to the distance statistic, the OPs of $ \mathrm{U}_{1}$ and $\mathrm{U}_{2}$ are given by 
\begin{equation}\label{Pin-U1-NOMA}
\begin{split}
 \mathbb{P}_{1}^{Pin,NOMA}&=\Pr( R_{1}^{Pin,NOMA}<\bar{R}),
\end{split}
\end{equation}
and
\begin{equation}\label{Pin-U2-NOMA}
\begin{split}
\mathbb{P}_{2}^{Pin,NOMA}&=\Pr( R_{2}^{Pin,NOMA}<\bar{R}),
\end{split}
\end{equation}
respectively.
\begin{theorem}\label{theorem-pin-noma}
In PASS with NOMA, the OPs of $ \mathrm{U}_{1}$ and $\mathrm{U}_{2}$ are given by
\begin{equation}\label{Pin-U1-NOMA-OP}
\begin{split}
 \mathbb{P}_{1}^{Pin,NOMA}&=1-\mathbb{I}_{4},
\end{split}
\end{equation}
and
\begin{equation}\label{Pin-U2-NOMA-OP}
\begin{split}
 \mathbb{P}_{2}^{Pin,NOMA}&= 1-(\mathbb{J}_{4} +\mathbb{Q}_{4}  + \mathbb{K}_{4} +\mathbb{C}_{4} +\mathbb{V}_{4}),
\end{split}
\end{equation}
respectively, where $\mathbb{I}_{4}=F_{Z_{3}}(b)$, $\mathbb{J}_{4} = \sum_{i=1}^{n} \omega_{i} \mathbf{j}_{4}(i)$, $\mathbb{Q}_{4} = \sum_{i=1}^{n} \omega_{i} \mathbf{q}_{4}(i)$, $\mathbb{K}_{4} = \sum_{i=1}^{n} \omega_{i} \mathbf{k}_{4}(i)$, $\mathbb{C}_{4} = \sum_{i=1}^{n} \omega_{i} \mathbf{c}_{4}(i)$, $\mathbb{V}_{4} = \sum_{i=1}^{n} \omega_{i} \mathbf{v}_{4}(i)$, $n$ is the number of nodes for the Chebyshev-Gauss quadrature, $\omega_i = \frac{\pi}{n}$ is the weight, $t_i = \cos \left( \frac{2i - 1}{2n} \pi \right)$, 
$\mathbf{j}_{4}(t) = \frac{1}{8} \sqrt{1 - t^{2}} e ^{-\varpi_2 \left(\frac{D^{2}t}{8}+d^{2}+\frac{D^{2}}{8}\right) ^{\frac{\alpha}{2}} }\frac{\sqrt{t + 1}}{2\sqrt{2}}$, 
$\mathbf{q}_{4}(t) = \frac{3}{16} \sqrt{1 - t^{2}} e ^{-\varpi_2\left(\frac{3D^{2}t}{8}+d^{2}+\frac{5D^{2}}{8}\right)^{\frac{\alpha}{2}} }$, 
$\mathbf{k}_{4}(t) = \frac{1}{8} \sqrt{1 - t^{2}} e^{-\varpi_2\left(\frac{D^{2}t}{8}+d^{2}+\frac{9D^{2}}{8}\right)^{\frac{\alpha}{2}} } \Big(\frac{1}{2}-\frac{\sqrt{t + 1}}{\sqrt{2}} +2\arctan\left(\frac{\sqrt{t + 1}}{2\sqrt{2}}\right)\Big)$,
$c_{4}(t) =\frac{11\sqrt{1 - t^{2}}e ^{-\varpi_2\left(\frac{11D^{2}t}{8}+d^{2}+\frac{21D^{2}}{8}\right)^{\frac{\alpha}{2}} }}{8} \\ \times \Big(\frac{2\pi - 1}{2}-2\arctan\left(\sqrt{\frac{11t + 19}{2}}\right)\Big)$, 
$v_{4}(t) =  \frac{e ^{-\varpi_2\left(\frac{D^{2}t}{8}+d^{2}+\frac{33D^{2}}{8}\right)^{\frac{\alpha}{2}}}}{8}\Big(-\frac{1}{2}+\frac{\sqrt{t + 1}}{2\sqrt{2}}+2\big(\arctan\left(\sqrt{\frac{2}{t + 31}}\right)\\ -\arctan\left(\frac{\sqrt{t + 1}}{4\sqrt{2}}\right)\big)\Big)\sqrt{1 - t^{2}}$, and $ \varpi_{2}= \frac{2^{\bar{R}}-1}{\rho(\alpha_{2}-\alpha_{1}(2^{\bar{R}}-1))}$.
\end{theorem}

\begin{proof}\label{Pin-proof2}
See appendix \ref{Appen-thro-4}.
\end{proof}

\begin{proposition}\label{proposition-pin-noma}
In the high-SNR regime, $\mathbb{P}_{1}^{Pin,NOMA}$ can be approximated as
\begin{equation}\label{eq_pin_noma_high-snr-u1}
    \begin{split}
        \mathbb{P}_{1}^{Pin,NOMA,\infty} = 0,
    \end{split}
\end{equation} 
and $\mathbb{P}_{2}^{Pin,NOMA}$ can be approximated as
\begin{equation}\label{eq_pin_noma_high-snr-u2}
    \begin{split}
        \mathbb{P}_{2}^{Pin,NOMA,\infty} = \tilde{c}_4\rho^{-1},
    \end{split}
\end{equation} 
where $\tilde{c}_4 = \int_{0}^{\infty}
            \omega_{2}z^{\frac{\alpha}{2}}f_{Z_{4}}(z)dz$.
Specifically, $\mathbb{P}_{2}^{Pin,NOMA}$ becomes zero for $\rho \geq \frac{(\frac{D^2}{2}+d^2)(2^{\bar{R}}-1)}{\alpha_1\eta}$.
\end{proposition}

\begin{proof}\label{pin-NOMA-high-SNR-proof}
    When $\rho \to \infty$, we know that $F_{Z_{3}}(\infty) = 1$, we have $\mathbb{P}_{1}^{Pin,NOMA,\infty} = 0$. By expanding the exponential function in $\mathbb{P}_{2}^{Pin,NOMA}= 1 - \int_{0}^{\infty} e^{-\varpi_{2} z^{\frac{\alpha}{2}}} f_{Z_{4}}(z) dz$, we have
    \begin{equation}\label{pin-NOMA-high-SNR-proof-r2}
        \begin{split}
            \mathbb{P}_{2}^{Pin,NOMA} &= 1 - \sum_{n=0}^{\infty}\rho^{-n}\int_{0}^{\infty}
            \frac{(-\omega_{2})^n}{n!}z^{\frac{n\alpha}{2}}f_{Z_{4}}(z)dz \\ 
            &= \sum_{n=1}^{\infty}\rho^{-n}\int_{0}^{\infty}
            \frac{(-1)^{n+1}\omega_{2}^n}{n!}z^{\frac{n\alpha}{2}}f_{Z_{4}}(z)dz  
        \end{split}
    \end{equation}
After extracting the dominant term, we can obtain \eqref{eq_pin_noma_high-snr-u2}.
The proof is complete.
\end{proof}

\begin{corollary}\label{corollary-pin-noma}
In the considered PASS with NOMA, the diversity order of $\mathrm{U}_2$ is given by $\mathcal{D}_{2}^{Pin,NOMA} = 1$.
\end{corollary}

\subsection{Performance Comparison}
\subsubsection{OMA}
According to the \eqref{Conv-U1-OMA-OP}, \eqref{eq_conv-oma-high-snr-u2}, \eqref{Pin-U1-OMA-OP} and \eqref{eq_pin-oma-high-snr-u2}, the performance gain of PASS compared to CASS with OMA is given by
\begin{equation}
    \begin{split}
    &\Delta_1^{OMA} = \mathbb{P}_{1}^{Conv,OMA} - \mathbb{P}_1^{Pin,OMA} \\
    &=
    \begin{cases}
        0 ,&z < d^2 ,\\ 
        \frac{2\sqrt{\zeta}}{D} - \frac{\pi \zeta}{D^2} ,& d^2 \leq z < d^2 + \frac{D^2}{4}, \\ 
        1 - \left(\frac{1}{D^2}\left(\hbar - \pi\zeta \right)\right) ,& d^2 + \frac{D^2}{4} \leq z < d^2 + \frac{D^2}{2} ,\\
        0 ,&  d^2 + \frac{D^2}{2} \leq z,
    \end{cases}
    \end{split}
\end{equation}
and 
\begin{equation}
\begin{split}
\Delta_2^{OMA,\infty}=& \mathbb{P}_{2}^{Pin,OMA,\infty} - \mathbb{P}_2^{Conv,OMA,\infty}
= \tilde{c}_5 \rho^{-1}  \approx 0,
    \end{split}
\end{equation}
where $\tilde{c}_5 = \tilde{c}_3-\tilde{c}_1$.

\begin{corollary}\label{corollary-delta-oma}
For OMA, the performance of $\mathrm{U}_1$ in the PASS is significantly enhanced compared with the CASS, while the performance of $\mathrm{U_2}$ is degraded slightly, especially in the high-SNR regime.
\end{corollary}

\begin{proof}\label{proof-Delta}
Obviously, for $\mathrm{U}_1$, every term of $\Delta_1^{OMA}$ is greater than or equal to zero. 
For $\mathrm{U}_2$ in the high-SNR regime, when $\rho \to \infty$, we have $\Delta_2^{OMA,\infty} \to 0$. The proof is complete. 
\end{proof}

\subsubsection{NOMA}
According to the \eqref{Conv-U1-NOMA-OP}, \eqref{eq_conv-noma-high-snr-u2}, \eqref{Pin-U1-NOMA-OP} and \eqref{eq_pin_noma_high-snr-u2}, the performance gain of PASS compared to CASS with NOMA is given by
\begin{equation}
    \begin{split}
    &\Delta_1^{NOMA} = \mathbb{P}_{1}^{Conv,NOMA} - \mathbb{P}_1^{Pin,NOMA}\\
    &=
    \begin{cases}
        0, &z < d^2 ,\\ 
        \frac{2\sqrt{\zeta}}{D} - \frac{\pi \zeta}{D^2} ,& d^2 \leq z < d^2 + \frac{D^2}{4} ,\\ 
        1 - (\frac{1}{D^2}\left(\hbar - \pi\zeta \right)) ,& d^2 + \frac{D^2}{4} \leq z < d^2 + \frac{D^2}{2} ,\\
        0 ,&  d^2 + \frac{D^2}{2} \leq z,
    \end{cases}
    \end{split}
\end{equation}
and 
\begin{equation}
\begin{split}
\Delta_2^{NOMA,\infty}&=\mathbb{P}_{2}^{Pin,NOMA,\infty} - \mathbb{P}_2^{Conv,NOMA,\infty}\\
&= \tilde{c}_6 \rho^{-1} \approx 0,
\end{split}
\end{equation}
where $\tilde{c}_6 = \tilde{c}_4 - \tilde{c}_2$.

\begin{corollary}\label{corollary-delta-noma}
For NOMA, the performance of $\mathrm{U}_1$ in the PASS is significantly enhanced compared with the CASS, while the performance of $\mathrm{U_2}$ is degraded slightly, especially in the high-SNR regime.
\end{corollary}

\begin{proof}
Similar to the proof of Corollary~\ref{corollary-delta-oma}.
\end{proof}

\section{Numerical Results}\label{Sec_Numerical_Results}
In this section, numerical results are presented to evaluate system performance. Meanwhile, Monte-Carlo simulations are conducted to verify the accuracy. The results reflect the superiority of PASS over CASS in terms of outage probability. Specifically, the parameters are set as shown in Table~\ref{Tab_Parameters}. Specifically, $n$ is a trade-off parameter of accuracy and complexity. The approximation accuracy increases as $n$ increases. By numerical simulations, we find that the setting $n = 100$ can generally achieve an approximation error of $10^{-10}$ which is negligible.
\begin{table}
\centering
\caption{Parameters setting}
\begin{tabular}{|p{0.24\textwidth}|p{0.16\textwidth}|}
\hline
Side length of the room & $D = 20m$ or $30m$\\
\hline
Height of the waveguide& $d = 5m$\\
\hline
Carrier frequency& $f_{c} = 10\ GHz$\\
\hline
Bandwidth& $B = 1\ MHz$\\
\hline
Path-loss exponent& $\alpha = 6$\\
\hline
The power allocation coefficients for $\mathrm{U}_1$ and $\mathrm{U}_2$ (NOMA)&  $\alpha_1 = 0.1$ and $\alpha_2 = 0.9$\\
\hline
Target data-rate & $\bar{R} = 1\ Mbps$\\
\hline
Number of points for Chebyshev-Gauss quadratures& $n=100$\\
\hline
\end{tabular}\label{Tab_Parameters}
\end{table}

\subsection{OPs of the CASS and PASS}
\begin{figure}
\centering
\includegraphics[width=\linewidth]{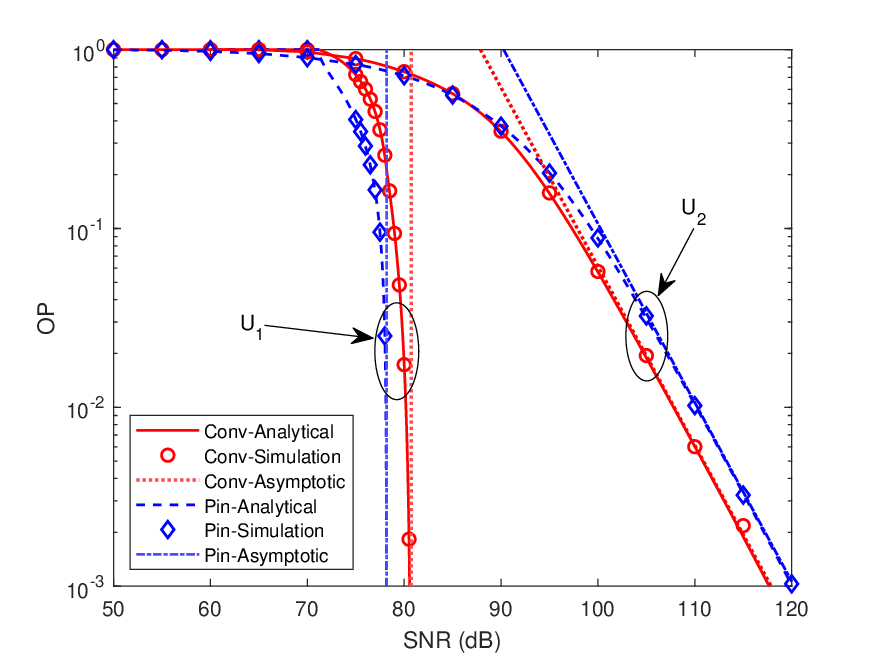}
\caption{OPs of the CASS and PASS with OMA.}
\label{fig_oma}
\end{figure}

In Fig.~\ref{fig_oma}, the OPs versus the transmit SNR in the CASS and PASS with OMA are plotted, where the CASS is regarded as the benchmark.
First, the OP analyses of $\mathrm{U}_1$ and $\mathrm{U}_2$ for OMA are accurate, since all simulation results coincide with the corresponding analytical results in Theorem~\ref{theorem-conv-oma} and Theorem~\ref{theorem-pin-oma}.
Next, it is observed that the OPs of $\mathrm{U}_1$ in the CASS and PASS with OMA become zero for $\text{SNR}\geq 81\ \text{dB}$ and $78\ \text{dB}$, respectively, as analyzed in Proposition~\ref{proposition-conv-oma} and Proposition~\ref{proposition-pin-oma}.
Then, it is observed that both OPs of $\mathrm{U}_2$ gradually approach their respective asymptotic curves derived in Proposition~\ref{proposition-conv-oma} and Proposition~\ref{proposition-pin-oma}, which validates our analysis. Furthermore, we find that the diversity orders of $\mathrm{U}_2$ in these two systems are one by observing the slopes, which is consistent with Corollary~\ref{corollary-conv-oma} and Corollary \ref{corollary-pin-oma}.

\begin{figure}
\centering
\includegraphics[width=\linewidth]{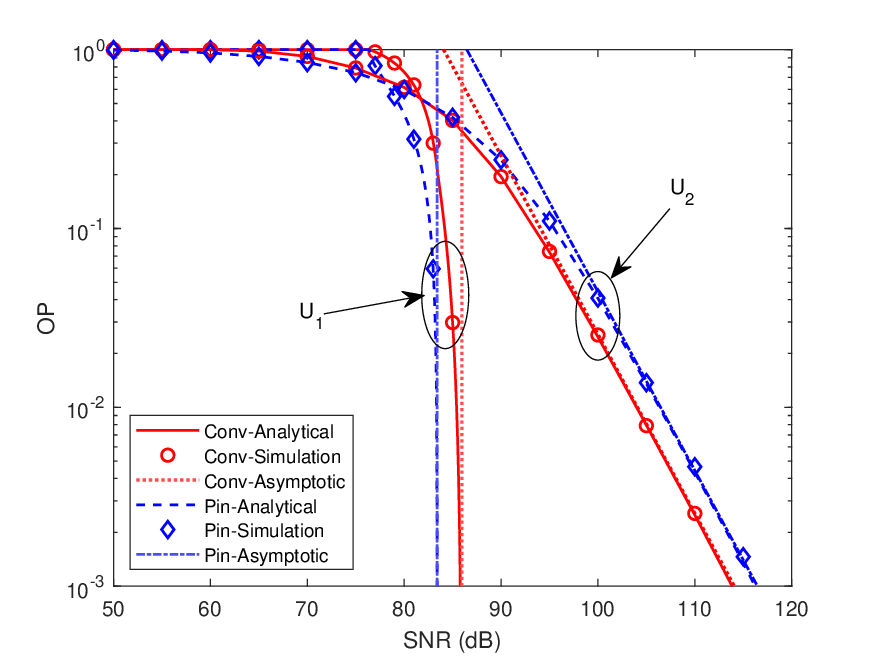}
\caption{OPs of the CASS and PASS with NOMA.}
\label{fig_noma}
\end{figure}

In Fig.~\ref{fig_noma}, the OPs versus the transmit SNR in the CASS and PASS with NOMA are plotted, where the CASS is regarded as the benchmark.
First, the OP analyses of $\mathrm{U}_1$ and $\mathrm{U}_2$ for NOMA are accurate, since all simulation results coincide with the corresponding analytical results in Theorem~\ref{theorem-conv-noma} and Theorem~\ref{theorem-pin-noma}.
Next, it is observed that the OPs of $\mathrm{U}_1$ in the CASS and PASS with NOMA become zero for $\text{SNR}\geq 86\ \text{dB}$ and $83\ \text{dB}$, respectively, as analyzed in Proposition~\ref{proposition-conv-noma} and Proposition~\ref{proposition-pin-noma}.
Then, it is observed that both OPs of $\mathrm{U}_2$ gradually approach their respective asymptotic curves derived in Proposition~\ref{proposition-conv-noma} and Proposition~\ref{proposition-pin-noma}, which validates our analysis. Furthermore, we find that the diversity orders of $\mathrm{U}_2$ in these two systems are one by observing the slopes, which is consistent with Corollary~\ref{corollary-conv-noma} and Corollary \ref{corollary-pin-noma}.

\begin{figure}
\centering
\includegraphics[width=\linewidth]{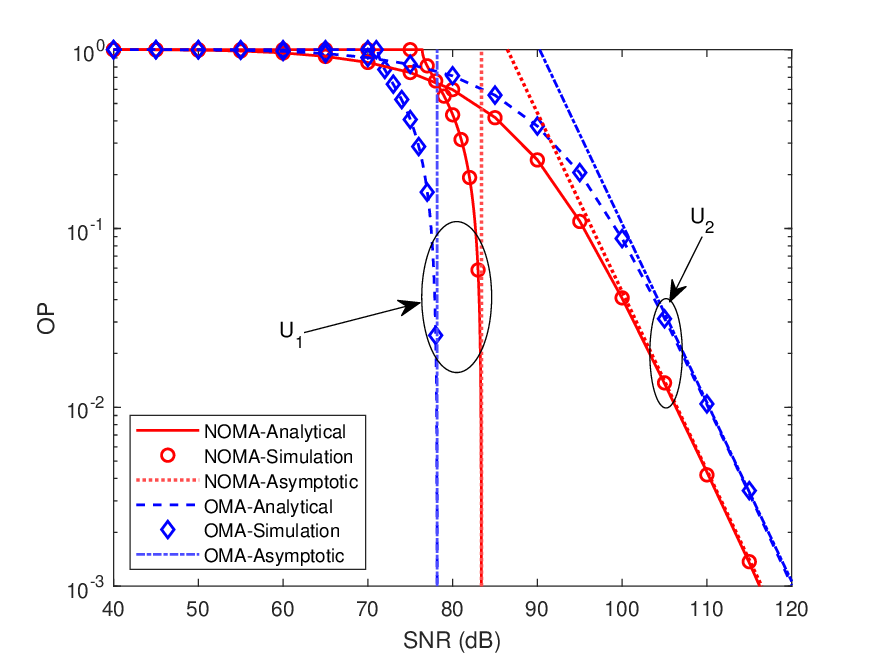}
\caption{OP comparison for OMA and NOMA in the PASS.}
\label{fig_pin-oma-noma}
\end{figure}

In Fig.~\ref{fig_pin-oma-noma}, the OPs versus the transmit SNR in the PASS with OMA and NOMA are plotted for comparisons. 
We observe that $\mathrm{U}_1$ in the NOMA system has worse performance than that in the OMA system. On the other hand, $\mathrm{U}_2$ in the NOMA system is observed to perform better than in the OMA system. 
It demonstrates that NOMA still has better fairness than OMA in PASS, which is consistent with the previous study of NOMA in CASS.

\subsection{The Performance Comparison of Single User for OMA}
\begin{figure}
\centering
\includegraphics[width=\linewidth]{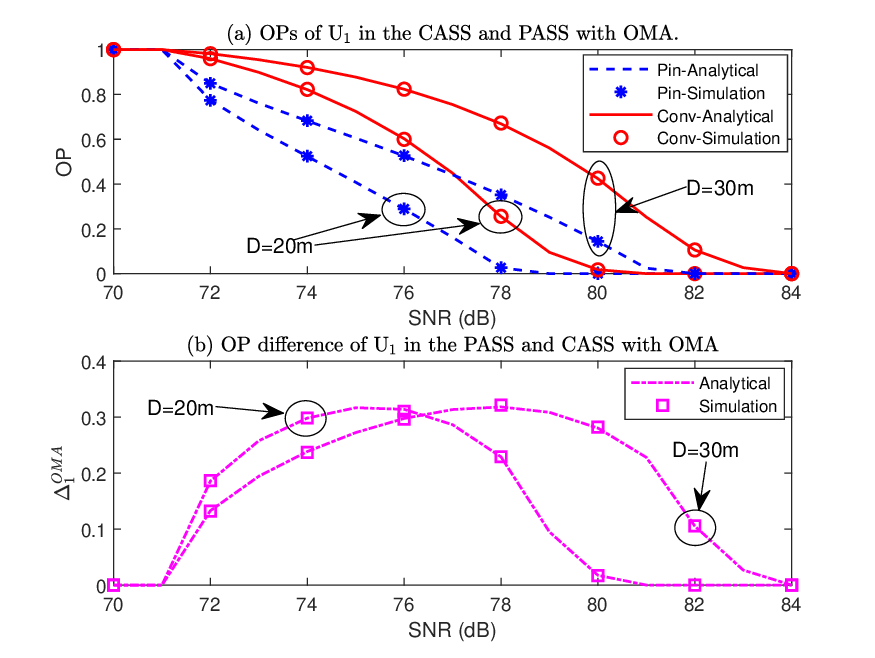}
\caption{The difference of $\mathrm{U}_1$'s OPs in the CASS and PASS with OMA.}
\label{fig_u1-oma}
\end{figure}

In Fig.~\ref{fig_u1-oma}(a), it is observed that both OPs of $\mathrm{U}_1$ in the CASS and PASS with OMA gradually approach zero derived from \eqref{eq_conv-oma-high-snr-u1} and \eqref{eq_pin-oma-high-snr-u1}, which validates our analysis.
We also observe that the OPs increase as the side length of the room increases. The reason is that the increase in side length makes the average channel gain of users worse. 
Furthermore, to highlight the performance difference of $\mathrm{U}_1$ in the CASS and PASS, we plot the difference of OPs in Fig.~\ref{fig_u1-oma}(b). It is observed that $\Delta_1^{OMA}=\mathbb{P}_{1}^{Conv,OMA}-\mathbb{P}_{1}^{Pin,OMA}$ is always greater than zero and is maximized in the middle SNR regime.
The change of the side length of the room only makes the difference shift along the $x$-axis.
This reflects the superiority of PASS over CASS.

\begin{figure}
\centering
\includegraphics[width=\linewidth]{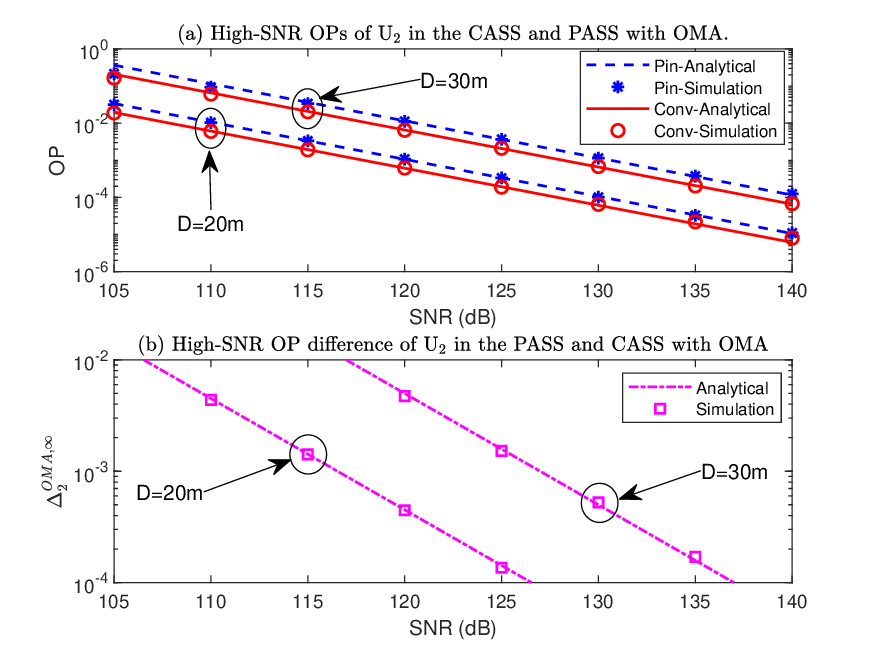}
\caption{The difference of the high-SNR approximations of $\mathrm{U}_2$'s OPs in the CASS and PASS with OMA.}
\label{fig_u2-oma}
\end{figure}

We further plot the high-SNR approximation curves of $\mathrm{U}_2$'s OPs in the CASS and PASS with OMA in Fig.~\ref{fig_u2-oma}(a). It is observed that both OPs gradually approach their respective asymptotic curves derived from \eqref{eq_conv-oma-high-snr-u2} and \eqref{eq_pin-oma-high-snr-u2}.
We also observe that the OPs increase as the side length of the room increases. The reason is that the increase in side length makes the average channel gain of users worse.
In Fig.~\ref{fig_u2-oma}(b), by observing the difference of $\mathrm{U}_2$'s OPs in the CASS and PASS in the high-SNR regime, i.e., $\Delta_2^{OMA,\infty}=\mathbb{P}_{2}^{Pin,OMA,\infty}-\mathbb{P}_{2}^{Conv,OMA,\infty}$, we find that the difference is miniscule in the high-SNR regime. We also observe that the change of the side length of the room only makes the difference shift along the $x$-axis. Thereby, we can conclude that for OMA, the performance of $\mathrm{U}_1$ in the PASS is significantly enhanced compared with the CASS, while the performance of $\mathrm{U_2}$ is slightly degraded, which is consistent with Corollary~\ref{corollary-delta-oma}. 

\subsection{The Performance Comparison of Single User for NOMA}
\begin{figure}
\centering
\includegraphics[width=\linewidth]{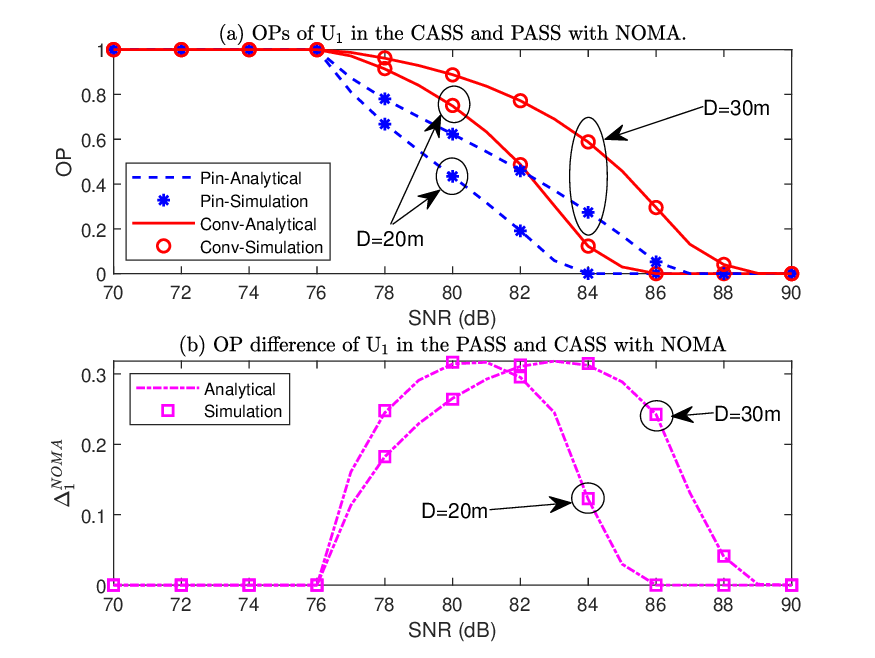}
\caption{The difference of $\mathrm{U}_1$'s OPs in the CASS and PASS with NOMA.}
\label{fig_u1-noma}
\end{figure}

In Fig.~\ref{fig_u1-noma}(a), it is observed that both OPs of $\mathrm{U}_1$ in the CASS and PASS with NOMA gradually approach zero derived from \eqref{eq_conv-noma-high-snr-u1} and \eqref{eq_pin_noma_high-snr-u1}. We also observe that the OPs increase as the side length of the room increases. The reason is that the increase in side length makes the average channel gain of users worse. 
Furthermore, to highlight the performance difference of $\mathrm{U}_1$ in the CASS and PASS, we plot the difference of OPs in Fig.~\ref{fig_u1-noma}(b). It is observed that $\Delta_1^{NOMA}=\mathbb{P}_{1}^{Conv,NOMA}-\mathbb{P}_{1}^{Pin,NOMA}$ is always greater than zero and is maximized in the middle SNR regime. The change of the side length of the room only makes the difference shift along the $x$-axis. This reflects the superiority of PASS over CASS.

\begin{figure}
\centering
\includegraphics[width=\linewidth]{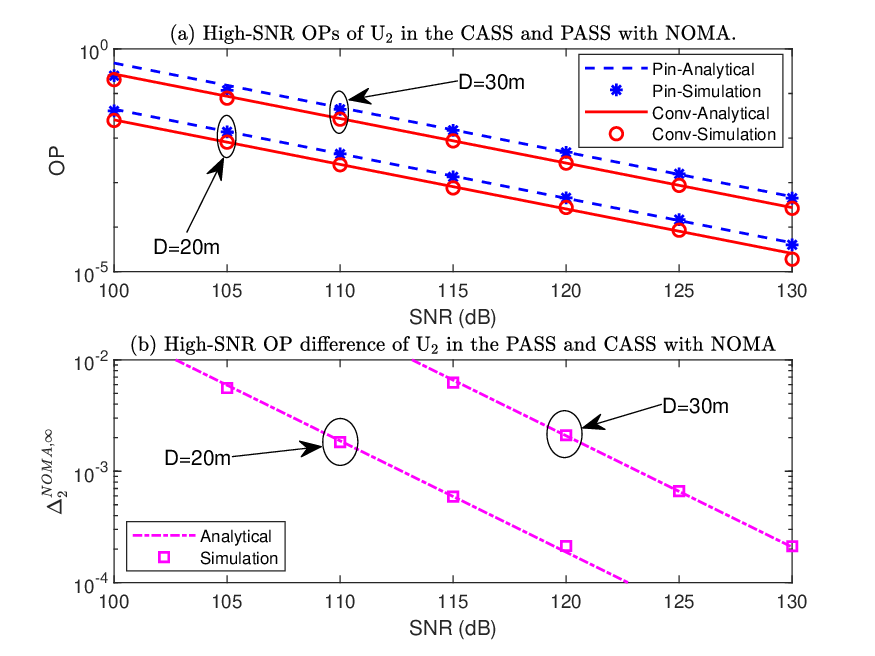}
\caption{The difference of the high-SNR approximations of $\mathrm{U}_2$'s OPs in the CASS and PASS with NOMA.}
\label{fig_u2-noma}
\end{figure}

We further plot the high-SNR approximation curves of $\mathrm{U}_2$'s OPs in the CASS and PASS with NOMA in Fig.~\ref{fig_u2-noma}(a). It is observed that both OPs of $\mathrm{U}_2$ for NOMA gradually approach their respective asymptotic curves derived from \eqref{eq_conv-noma-high-snr-u2} and \eqref{eq_pin_noma_high-snr-u2}.
We also observe that the OPs increase as the side length of the room increases. The reason is that the increase in side length makes the average channel gain of users worse.
In Fig.~\ref{fig_u2-noma}(b), by observing $\Delta_2^{NOMA,\infty}=\mathbb{P}_{2}^{Pin,NOMA,\infty}-\mathbb{P}_{2}^{Conv,NOMA,\infty}$, we find that the difference is diminishing in the high-SNR regime. We also observe that the change of the side length of the room only makes the difference shift along the $x$-axis. Thereby, we can conclude that for NOMA, the performance of $\mathrm{U}_1$ is significantly enhanced, while the performance of $\mathrm{U_2}$ is degraded slightly, which is consistent with Corollary~\ref{corollary-delta-noma}.

\section{Conclusion}\label{Sec_Conclusion}
This paper presented a comprehensive performance evaluation of CASS and PASS employing OMA and NOMA schemes.
Analytical derivations and numerical simulations consistently demonstrate that PASS achieve remarkable improvements in OP compared to CASS.
These gains underscore the adaptability and efficiency of PASS in overcoming propagation challenges, particularly in heterogeneous environments with mixed LoS/NLoS conditions. 
Future work could explore the optimal positions of PAs, the trade-off between the number of PAs and cost, and the performance analysis for multiple-antenna PASS, to further advance this technology.

\begin{appendices}
\section{Proof of Theorem \ref{theorem-conv-oma}}\label{Appen-thro-1}
\renewcommand{\theequation}{\thesection.\arabic{equation}}
\setcounter{equation}{0}
First, we denote $Y=|h|^{2}$, and its PDF and CDF are given by $f_{Y}(y)=e^{-y}$ for $y \geq 0$ and $F_{Y}(y)=1-e^{-y}$ for $ y \geq 0$, respectively. Then, $\mathbb{P}_{1}^{Conv,OMA}$ can be derived as 
\begin{equation}\label{proof-Conv-R1-part1}
\begin{split}
 \mathbb{P}_{1}^{Conv,OMA} 
    &=\Pr\left(\frac{1}{M}\log_2\left(1+\frac{\eta \rho} {Z_1}\right)<\bar{R}\right)\\
    &=1-F_{Z_{1}}\left(\frac{\eta \rho}{2^{M\bar{R}}- 1}\right).
\end{split}
\end{equation}
Based on \eqref{CDF-Z1}, $\mathbb{P}_{1}^{Conv,OMA} $ is derived.

On the other hand, $\mathbb{P}_{2}^{Conv,OMA}$ can be derived as
\begin{equation}\label{proof-Conv-R2-part1}
\begin{split}
\mathbb{P}_{2}^{Conv,OMA}
&=\Pr \left( \frac{1}{M} \log_2 \left( 1 + \frac{\rho |h|^2}{Z_{2}^{\frac{\alpha}{2}}} \right) < \bar{R} \right)\\
&= \int_{0}^{\infty} \left( 1 - e^{\frac{-\left(2^{M \bar{R}} - 1\right) Z^{\frac{\alpha}{2}}}{\rho}} \right) f_{Z_2}(z) dz.
\end{split}
\end{equation}
Let $\varpi_{1}=\frac{2^{M\bar{R}}-1}{\rho}$, and we can derive $\mathbb{P}_{2}^{Conv,OMA}$ by \eqref{proof-Conv-R2-part2}.
\begin{figure*}[]
\begin{equation}\label{proof-Conv-R2-part2}
\begin{split}
& \mathbb{P}_{2}^{Conv,OMA}= 1 - \int_{0}^{\infty} e^{-\varpi_{1} z^{\frac{\alpha}{2}}} f_{Z_{2}}(z) dz
=1- \frac{1}{D^2}\underbrace{ \int_{d^2 + \frac{D^2}{4}}^{d^2 + \frac{D^2}{2}} e^{-\varpi_{1} z^{\frac{\alpha}{2}}} \left( \frac{\pi}{2} - \arcsin \left( \frac{D}{2\sqrt{z - d^2}} \right) \right) dz}_{\mathbb{J}_{1}}\\
&- \frac{1}{D^2}\underbrace{\int_{d^2 + \frac{D^2}{2}}^{d^2 + \frac{9D^2}{4}} e^{-\varpi_{1} z^{\frac{\alpha}{2}}} \left( \frac{\pi}{2} - \arcsin \left( \sqrt{1 - \frac{D^2}{4(z - d^2)}} \right) \right) dz}_{\mathbb{Q}_{1}}\\
&-\frac{1}{D^2} \underbrace{\int_{d^2 + \frac{9D^2}{4}}^{d^2 + \frac{5D^2}{2}} e^{-\varpi_{1} z^{\frac{\alpha}{2}}} \left( \arcsin \left( \frac{3D}{2\sqrt{z - d^2}} \right) + \arcsin \left( \frac{D}{2\sqrt{z - d^2}} \right) - \frac{\pi}{2} \right) dz}_{\mathbb{K}_{1}}.
\end{split}
\end{equation}
\end{figure*}
To perform Gaussian-Chebyshev numerical integration for each integration, we need a variable substitution to change the integration domain to [-1,1], i.e., letting $z = \frac{D^2}{8} t + \frac{3D^2}{8}$,
and we have
\begin{equation}\label{proof-Conv-R2-part3}
\begin{split}
\mathbb{J}_{1} 
     &= \int_{-1}^{1} \frac{D^2}{8} e^{-\varpi_{1} (\frac{D^2}{8} t + \frac{3D^2}{8})^{\frac{\alpha}{2}}} \left( \frac{\pi}{2} - \arcsin \left( \sqrt{\frac{2}{t + 3}} \right) \right) dt \\
  &= \sum_{i=1}^{n} \omega_i \mathbf{j}_{1}(t_i).
\end{split}
\end{equation}
Similarly, $\mathbb{Q}_{1} $ and $\mathbb{K}_{1}$ can be solved in the same way.
The proof is complete.

\section{Proof of Theorem \ref{theorem-conv-noma}}\label{Appen-thro-2}
\renewcommand{\theequation}{\thesection.\arabic{equation}}
\setcounter{equation}{0}
$\mathbb{P}_{1}^{Conv,NOMA}$ can be derived as 
\begin{equation}\label{proof-R1-Conv-NOMA-part1}
\begin{split}
\mathbb{P}_{1}^{Conv,NOMA}
    &=\Pr\left(\log_2\left(1+\frac{\alpha_{1}\eta \rho}{Z_{1}}\right)<\bar{R}\right)\\
    &=1-F_{Z_{1}}\left(\frac{\alpha_{1}\eta \rho}{2^{\bar{R}}-1}\right).
\end{split}
\end{equation}
Based on \eqref{CDF-Z1}, $\mathbb{P}_{1}^{Conv,NOMA}$ is derived.

On the other hand, $\mathbb{P}_{2}^{Conv,NOMA}$ can be derived as
\begin{equation}\label{proof-R2-NOMA-part1}
\begin{split}
\mathbb{P}_{2}^{Conv,NOMA}
&=\Pr \left(\log_2\left(1+\frac{\alpha_{2}\rho|h_{2}|^{2}z^{-\alpha/2}}{\alpha_{1}\rho|h_{2}|^{2}z^{-\alpha/2}+1}\right) < \bar{R} \right) \\
&= \Pr \left(  y <\frac{(\alpha_{2}-\alpha_{1}(2^{\bar{R}}-1))\rho}{(2^{\bar{R}}-1)z^{\alpha/2}}\right).
\end{split}
\end{equation}
Let $ \varpi_{2}= \frac{2^{\bar{R}}-1}{\rho(\alpha_{2}-\alpha_{1}(2^{\bar{R}}-1))}$, and we can derive $\mathbb{P}_{2}^{Conv,NOMA}$ by \eqref{proof-R2-NOMA-part2}.
\begin{figure*}[]
	\begin{equation}\label{proof-R2-NOMA-part2}
		\begin{split}
			& \mathbb{P}_{2}^{Conv,NOMA}
            = \int_{0}^{\infty} \int_{0}^{\varpi_{2} z^{\alpha/2}}f_{Z_{2}}(z) f_Y(y) dy dz
			=1- \frac{1}{D^2}\underbrace{ \int_{d^2 + \frac{D^2}{4}}^{d^2 + \frac{D^2}{2}} e^{-\varpi_{2} z^{\frac{\alpha}{2}}} \left( \frac{\pi}{2} - \arcsin \left( \frac{D}{2\sqrt{z - d^2}} \right) \right) dz}_{\mathbb{J}_{2}}\\
            &- \frac{1}{D^2}\underbrace{\int_{d^2 + \frac{D^2}{2}}^{d^2 + \frac{9D^2}{4}} e^{-\varpi_{2} z^{\frac{\alpha}{2}}} \left( \frac{\pi}{2} - \arcsin \left( \sqrt{1 - \frac{D^2}{4(z - d^2)}} \right) \right) dz}_{\mathbb{Q}_{2}}\\
			&-\frac{1}{D^2} \underbrace{ \int_{d^2 + \frac{9D^2}{4}}^{d^2 + \frac{5D^2}{2}} e^{-\varpi_{2} z^{\frac{\alpha}{2}}} \left( \arcsin \left( \frac{3D}{2\sqrt{z - d^2}} \right) + \arcsin \left( \frac{D}{2\sqrt{z - d^2}} \right) - \frac{\pi}{2} \right) dz}_{\mathbb{K}_{2}}.
		\end{split}
	\end{equation}
\end{figure*}
To perform Gaussian-Chebyshev numerical integration for each integration, we need a variable substitution to change the integration domain to [-1,1], i.e., letting $z = \frac{D^2}{8} t + \frac{3D^2}{8}$,
and we have
\begin{equation}\label{proof-R2-part3}
\begin{split}
\mathbb{J}_{2}
     &= \int_{-1}^{1} \frac{D^2}{8} e^{-\varpi_{2} (\frac{D^2}{8} t + \frac{3D^2}{8})^{\frac{\alpha}{2}}} \left( \frac{\pi}{2} - \arcsin \left( \sqrt{\frac{2}{t + 3}} \right) \right) dt \\
  &= \sum_{i=1}^{n} \omega_i \mathbf{j}_{2}(t_i).
\end{split}
\end{equation}
Similarly, $\mathbb{Q}_{2} $ and $\mathbb{K}_{2}$ can be solved in the same way.
The proof is complete.

\section{Proof of Theorem \ref{theorem-pin-oma}}\label{Appen-thro-3}
\renewcommand{\theequation}{\thesection.\arabic{equation}}
\setcounter{equation}{0}
First, we denote $Y=|h|^{2}$, and its PDF and CDF are given by $f_{Y}(y)=e^{-y}$ for $y \geq 0$ and $F_{Y}(y)=1-e^{-y}$ for $y \geq 0$, respectively. Then, $\mathbb{P}_{1}^{Conv,OMA}$ can be derived as 
\begin{equation}\label{proof-Pin-R1-part1}
\begin{split}
 \mathbb{P}_{1}^{Pin,OMA} 
    &=\Pr\left(\frac{1}{M}\log_2\left(1+\frac{\eta \rho}{Z_{3}}\right)<\bar{R}\right)\\
    &=1-F_{Z_{3}}\left(\frac{\eta \rho}{2^{M\bar{R}}-1}\right).
\end{split}
\end{equation}
Based on \eqref{CDF-z3}, $\mathbb{P}_{1}^{Pin,OMA} $ is derived.

On the other hand, $\mathbb{P}_{2}^{Pin,OMA}$ can be derived as
\begin{equation}\label{proof-Pin-R2-part1}
\begin{split}
\mathbb{P}_{2}^{Pin,OMA}
    &=\Pr \left( \frac{1}{M} \log_2 \left( 1 + \frac{\rho |h|^2}{Z_{4}^{\frac{\alpha}{2}} } \right) < \bar{R} \right)\\
&= \int_{0}^{\infty} \left( 1 - e^{\frac{-(2^{M \bar{R}} - 1) Z^{\frac{\alpha}{2}} }{\rho}} \right) f_{Z_4}(z) dz.
\end{split}
\end{equation}
Let $\varpi_{1}=\frac{2^{M\bar{R}}-1}{\rho}$, and we can derive $\mathbb{P}_{2}^{Pin,OMA}$ by \eqref{proof-Pin-R2-part2}.
\begin{figure*}[]
	\begin{equation}\label{proof-Pin-R2-part2}
		\begin{split}
			& \mathbb{P}_{2}^{Pin,OMA}= 1 - \int_{0}^{\infty} e^{-\varpi_{1} z^{\frac{\alpha}{2}}} f_{Z_{4}}(z) dz
			=1- \underbrace{\int_{d^2}^{d^2 + \frac{D^2}{4}} e^{-\varpi_{1} z^{\frac{\alpha}{2}}}\frac{\sqrt{\zeta}}{D^3}}_{\mathbb{J}_{3}} - 
            \underbrace{\int_{d^2 + \frac{D^2}{4}}^{d^2 + D^2} e^{-\varpi_{1} z^{\frac{\alpha}{2}}}\frac{1}{2D^2} dz}_{\mathbb{Q}_{3}} \\
            & - \underbrace{\int_{d^2 + D^2}^{d^2 + \frac{5D^2}{4}} e^{-\varpi_{1} z^{\frac{\alpha}{2}}}  \frac{D-4\sqrt{\varrho}+4D\left(\arctan \left(\frac{\sqrt{\varrho}}{D} \right) -\frac{\sqrt{\varrho}}{D} + \frac{D\sqrt{\varrho}}{\zeta}\right)}{2D^3}  dz}_{\mathbb{K}_3}
			-\underbrace{\int_{d^2 + \frac{5D^2}{4}}^{d^2 + 4D^2} e^{-\varpi_{1} z^{\frac{\alpha}{2}}} \frac{2\pi-1-4\arctan \left(\frac{\sqrt{\varsigma}}{D} \right)}{2D^2} dz}_{\mathbb{C}_{3}} \\
            &- \underbrace{ \int_{d^2 + 4D^2}^{d^2 + \frac{17D^2}{4}} e^{-\varpi_{1} z^{\frac{\alpha}{2}}} \left( \frac{2}{D\sqrt{\varsigma}} - \frac{1}{2D^2}+\frac{\sqrt{\tau}}{D^3} + \frac{2}{D^2}\left(
              \left(\arctan \left(\frac{D}{\sqrt{\varsigma}} \right) - \arctan \left(\frac{\sqrt{\tau}}{2D} \right) + 
              \frac{7D\sqrt{\tau}}{4\zeta} \right)
              -\frac{D}{\sqrt{\varsigma}}-\frac{7D\sqrt{\tau}}{4\zeta} 
              \right) \right) dz}_{\mathbb{V}_3}.
		\end{split}
	\end{equation}
\end{figure*}
To perform Gaussian-Chebyshev numerical integration for each integration, we need a variable substitution to change the integration domain to [-1,1], i.e., letting $z = \frac{D^2}{8} t + \frac{D^2}{8} + d^2$,
and we have
\begin{equation}\label{proof-Pin-R2-part3}
\begin{split}
\mathbb{J}_{3} 
     &= \int_{-1}^{1} \frac{D^2}{8} e^{-\varpi_{1} (\frac{D^2}{8} t + \frac{D^2}{8})^{\frac{\alpha}{2}}} \left( \frac{\sqrt{t+1}}{2\sqrt{2}D} \right) dt \\
  &= \sum_{i=1}^{n} \omega_i \mathbf{j}_{3}(t_i).
\end{split}
\end{equation}
Similarly, $\mathbb{Q}_{3} $, $\mathbb{K}_{3}$, $\mathbb{C}_3$, and $\mathbb{V}_3$ can be solved in the same way.
The proof is complete.

\section{Proof of Theorem \ref{theorem-pin-noma}}\label{Appen-thro-4}
\renewcommand{\theequation}{\thesection.\arabic{equation}}
\setcounter{equation}{0}
$\mathbb{P}_{1}^{Pin,NOMA}$ can be derived as 
\begin{equation}\label{proof-R1-Pin-NOMA-part1}
\begin{split}
\mathbb{P}_{1}^{Pin,NOMA}    &=\Pr\left(\log_2\left(1+\frac{\alpha_{1}\eta \rho}{Z_{3}}\right)<\bar{R}\right)\\
    &=1-F_{Z_{3}}\left(\frac{\alpha_{1}\eta \rho}{2^{\bar{R}}-1}\right).
\end{split}
\end{equation}
Based on \eqref{CDF-z3}, $\mathbb{P}_{1}^{Pin,NOMA} $ is derived.

On the other hand, $\mathbb{P}_{2}^{Pin,NOMA}$ can be derived as
\begin{equation}\label{proof-R2-NOMA-pin-part1}
\begin{split}
\mathbb{P}_{2}^{Pin,NOMA}
    &=\Pr \left(\log_2\left(1+\frac{\alpha_{2}\rho|h_{2}|^{2}z^{-\alpha/2}}{\alpha_{1} \rho|h_{2}|^{2}z^{-\alpha/2}+1}\right) < \bar{R} \right) \\
&= \Pr \left( y < \frac{(\alpha_{2}-\alpha_{1}(2^{\bar{R}}-1))P_{b}}{\sigma^2(2^{\bar{R}}-1)z^{\alpha/2}}\right).
\end{split}
\end{equation}
Let $ \varpi_{2}= \frac{2^{\bar{R}}-1}{\rho(\alpha_{2}-\alpha_{1}(2^{\bar{R}}-1))}$, and we can derive $\mathbb{P}_{2}^{Pin,NOMA}$ by \eqref{proof-R2-NOMA-pin-part2}.
\begin{figure*}[]
\begin{equation}\label{proof-R2-NOMA-pin-part2}
\begin{split}
& \mathbb{P}_{2}^{Pin,NOMA}
= \int_{0}^{\infty} \int_{0}^{\varpi_{2} z^{\alpha/2}}f_{Z_{4}}(z) f_Y(y) dy dz =1- \underbrace{\int_{d^2}^{d^2 + \frac{D^2}{4}} e^{-\varpi_{2} z^{\frac{\alpha}{2}}}\frac{\sqrt{\zeta}}{D^3}dz}_{\mathbb{J}_{4}} - \underbrace{\int_{d^2 + \frac{D^2}{4}}^{d^2 + D^2} e^{-\varpi_{2} z^{\frac{\alpha}{2}}}\frac{1}{2D^2} dz}_{\mathbb{Q}_{4}} \\
& - \underbrace{\int_{d^2 + D^2}^{d^2 + \frac{5D^2}{4}} e^{-\varpi_{2} z^{\frac{\alpha}{2}}}  \frac{D-4\sqrt{\varrho}+4D\left(\arctan \left(\frac{\sqrt{\varrho}}{D} \right) -\frac{\sqrt{\varrho}}{D} + \frac{D\sqrt{\varrho}}{\zeta}\right)}{2D^3} dz}_{\mathbb{K}_4} -\underbrace{\int_{d^2 + \frac{5D^2}{4}}^{d^2 + 4D^2} e^{-\varpi_{2} z^{\frac{\alpha}{2}}}  \frac{2\pi-1-4\arctan \left(\frac{\sqrt{\varsigma}}{D} \right)}{2D^2}  dz}_{\mathbb{C}_{4}} \\
&- \underbrace{\int_{d^2 + 4D^2}^{d^2 + \frac{17D^2}{4}} e^{-\varpi_{2} z^{\frac{\alpha}{2}}} \left( \frac{2}{D\sqrt{\varsigma}} - \frac{1}{2D^2}+\frac{\sqrt{\tau}}{D^3} + \frac{2}{D^2}\left(\left(\arctan \left(\frac{D}{\sqrt{\varsigma}} \right) - \arctan \left(\frac{\sqrt{\tau}}{2D} \right) + \frac{7D\sqrt{\tau}}{4\zeta} \right) -\frac{D}{\sqrt{\varsigma}}-\frac{7D\sqrt{\tau}}{4\zeta} \right) \right) dz}_{\mathbb{V}_4}.
\end{split}
\end{equation}
\end{figure*}
To perform Gaussian-Chebyshev numerical integration for each integration, we need a variable substitution to change the integration domain to [-1,1], i.e., letting
$z = \frac{D^2}{8} t + \frac{3D^2}{8}$,
and we have
\begin{equation}\label{proof-R2-pin-part3}
\begin{split}
\mathbb{J}_{4} 
     &= \int_{-1}^{1} \frac{D^2}{8} e^{-\varpi_{2} \left(\frac{D^2}{8} t + \frac{D^2}{8}\right)^{\frac{\alpha}{2}}} \left( \frac{\sqrt{t+1}}{2\sqrt{2}D} \right) dt \\
  &= \sum_{i=1}^{n} \omega_i \mathbf{j}_{4}(t_i).
\end{split}
\end{equation}
Similarly, $\mathbb{Q}_{4} $, $\mathbb{K}_{4}$, $\mathbb{C}_4$, and $\mathbb{V}_4$ can be solved in the same way.
The proof is complete.
\end{appendices}

\bibliographystyle{IEEEtran}
\bibliography{ref}

\end{document}